\documentclass[onefignum,onetabnum]{siamonline220329}

\usepackage{lipsum}
\usepackage{amsfonts}
\usepackage{graphicx}
\usepackage{epstopdf}
\usepackage{algorithmic}
\ifpdf
  \DeclareGraphicsExtensions{.eps,.pdf,.png,.jpg}
\else
  \DeclareGraphicsExtensions{.eps}
\fi

\usepackage{enumitem}
\setlist[enumerate]{leftmargin=.5in}
\setlist[itemize]{leftmargin=.5in}

\newsiamremark{remark}{Remark}
\newsiamremark{hypothesis}{Hypothesis}
\crefname{hypothesis}{Hypothesis}{Hypotheses}
\newsiamthm{claim}{Claim}

\headers{Phylogenetic invariants under a Markov model}{M. Frohn, N. Holtgrefe, L. van Iersel, M. Jones, S. Kelk}

\title{Invariants for level-1 phylogenetic networks under the random walk 4-state Markov model\thanks{Submitted to the editors \today.
\funding{This work was supported by grants OCENW.M.21.306 and OCENW.KLEIN.125
from the Dutch Research Council (NWO)}}}

\author{Martin Frohn\thanks{Department of Advanced Computing Sciences, Maastricht University, Maastricht, The Netherlands
  (\email{martin.frohn@maastrichtuniversity.nl}, \email{steven.kelk@maastrichtuniversity.nl}).}
\and Niels Holtgrefe\thanks{Institute of Applied Mathematics, Delft University of Technology, Delft, The Netherlands
  (\email{N.A.L.Holtgrefe@tudelft.nl}, \email{l.j.j.vaniersel@tudelft.nl}, \email{M.E.L.Jones@tudelft.nl}).}
\and Leo van Iersel\footnotemark[3]
\and Mark Jones\footnotemark[3]
\and Steven Kelk\footnotemark[2]}

\usepackage{amsopn}

\newcommand{\leo}[1]{\textcolor{black}{#1}}
\usepackage{todonotes}

\ifpdf
\hypersetup{
  pdftitle={Invariants for level-1 phylogenetic networks under the random walk 4-state Markov model},
  pdfauthor={M. Frohn, N. Holtgrefe, L. van Iersel, M. Jones, and S. Kelk}
}
\fi

\begin{document}

\maketitle

\begin{abstract}
Phylogenetic networks can represent evolutionary events that cannot be described by phylogenetic trees, such as hybridization, introgression, and lateral gene transfer. Studying phylogenetic networks under a statistical model of DNA sequence evolution can aid the inference of phylogenetic networks. Most notably Markov models like the Jukes-Cantor or Kimura-3 model can been employed to infer a phylogenetic network using phylogenetic invariants. In this article we determine all quadratic invariants for sunlet networks under the random walk 4-state Markov model, which includes the aforementioned models. Taking toric fiber products of trees and sunlet networks, we obtain a new class of invariants for level-1 phylogenetic networks under the same model. Furthermore, we apply our results to the identifiability problem of a network parameter. In particular, we prove that our new class of invariants of the studied model is not sufficient to derive identifiability of quarnets ($4$-leaf networks). Moreover, we provide an efficient method that is faster and more reliable than the state-of-the-art in finding a significant number of invariants for many level-1 phylogenetic networks.
\end{abstract}



\section{Introduction}
The field of phylogenetics aims to infer evolutionary relationships between different organisms. To explain evolutionary phenomena a wide variety of mathematical models are studied depending on the available biological data and different assumptions on underlying evolutionary processes~\cite{Felsen04,Steel16,Steel03}. In this article we take an algebraic perspective of models in phylogenetics that was intiated in computational biology 37 years ago~\cite{CF87,Lake87}. Specifically, we will show how to find generators of the vanishing ideal of a network-based Markov model of DNA sequence evolution. Such generators are called \emph{phylogenetic invariants} and some of them were first discovered for the random walk 4-state Markov model by Evans and Speed~\cite{ES93} under the assumption that evolutionary histories follow a treelike process. In the phylogenetics literature, the four states of their Markov model are identified with letters of the alphabet $\Sigma=\{A,G,C,T\}$, representing the bases adenine (A), guanine (G), cytosine (C) and thymine (T) in the deoxyribonucleic acid (DNA) of an organism~\cite{Anfinsen59}. In this context, the random walk is taken on pairs of bases in the statespace~$\Sigma$ and not on the graph structure of a (tree-like) network. For example, Jukes and Cantor~\cite{JC69} introduced the first model for phylogenetic inference which fits into the class of models studied by Evans and Speed~\cite{ES93}. Moreover, random walk 4-state Markov models are statistical models with a rich history in phylogenetics that remains relevant today (see~\cite{K81,Yang14,Posada08}). Furthermore, the use of Fourier analysis in~\cite{ES93} to derive phylogenetic invariants inspired many others to derive further invariants for different tree-based Markov models~\cite{EZ98,FS95,Hage00}. These contributions lead to an article by Sturmfels and Sullivant~\cite{SS05} who showed that phylogenetic invariants of the tree-based random walk 4-state Markov model form a toric ideal in the Fourier coordinates. Moreover, the authors determined generators and Gr\"obner bases for these toric ideals. In this article we determine all quadratic phylogenetic invariants of the same model for a class of networks more general than trees which are known in the literature as level-1 phylogenetic networks~\cite{Kong22}. To this end, we will generalize the analysis of Cummings et al.~\cite{CHM24} who determined all quadratic phylogenetic invariants of sunlet networks under the random walk 2-state Markov model. \leo{We emphasize that the generalization from two to four states is highly nontrivial. The} 
binary model was introduced by Cavender, Farris and Neyman~\cite{Cavender78,Farris73,Neyman71} and assumes that mutations between both purines A and G, and pyrimidines C and T are much more common than any other substitutions in the DNA. From an algebraic perspective, this assumption significantly simplifies the study of phylogenetic invariants because a mutation between a purine and pyrimidine is symmetric and no other substitution event is present in the model. In our analysis we do not discard any interaction between the states in $\Sigma$. This generalization allows us to outline an efficient algorithm which can calculate a significant number of invariants for all level-1 phylogenetic networks under the random walk 4-state Markov model. Subsequently, we consider the question of identifiability of the network parameter under a statistical model. That is, can the network parameter be recovered from the observables of the model in a one-to-one mapping?  We prove that quadratic phylogenetic invariants are not sufficient to derive identifiability \leo{of level-$1$ quarnets ($4$-leaf networks). Therefore, identifiability of general level-$1$ networks can not be shown using quarnets as done by Gross et al.~\cite{GIJ21} with higher degree invariants.}
The state-of-the-art algorithm \leo{that Gross et al.~\cite{GIJ21}} used to certify identifiability under the random walk 4-state Markov model is randomized, relies on a matroid separation routine which can fail to prove identifiability for a given phylogenetic network and requires the use of a generic Gröbner basis calculation method. For a detailed description of the state-of-the-art see~\cite{HS21}. Hence, even though our algorithm calculates phylogenetic invariants of limited degree, it does not face the same methodological constraints and limitations as the state-of-the-art and speeds up the calculation of invariants for many level-1 phylogenetic networks.

The article is structured as follows: in Section~\ref{sec2} we define and detail the algebraic model we study in this article. For this model we determine all quadratic invariants of the sunlet network in Section~\ref{sec3}. In Section~\ref{sec4} we summarize our findings and discuss the implications of our results as well as future research directions.

\section{Notation and background}\label{sec2}
In this section we introduce the concepts from graph theory, phylogenetics and algebra which we make use of in this article.
\subsection{Phylogenetic networks}
For a natural number $n\geq 2$, let $\Gamma=\{1,2,\dots,n\}$ represent a set of distinct (evolving) objects, called \emph{taxa}. We call a connected graph $G$ \emph{rooted} if $G$ is directed and contains a unique vertex $\rho$ having in-degree zero and there exists a directed path from $\rho$ to any vertex of in-degree 1 and out-degree 0 of $G$. We call $\rho$ the \emph{root} of $G$. A vertex of $G$ with out-degree zero and in-degree one is called a \emph{leaf} of $G$.
\begin{definition}
A \emph{(binary) phylogenetic network} of $\Gamma$ is a rooted acyclic directed graph with no edges in parallel such that
\begin{enumerate}
\item the root has out-degree two;
\item the leaves are in bijection with the set $\Gamma$;
\item all other vertices, called \emph{internal vertices}, have either in-degree one and out-degree two, or in-degree two and out-degree one.
\end{enumerate}
\end{definition}
\begin{figure}[t]
  \centering
  \includegraphics[scale=0.5]{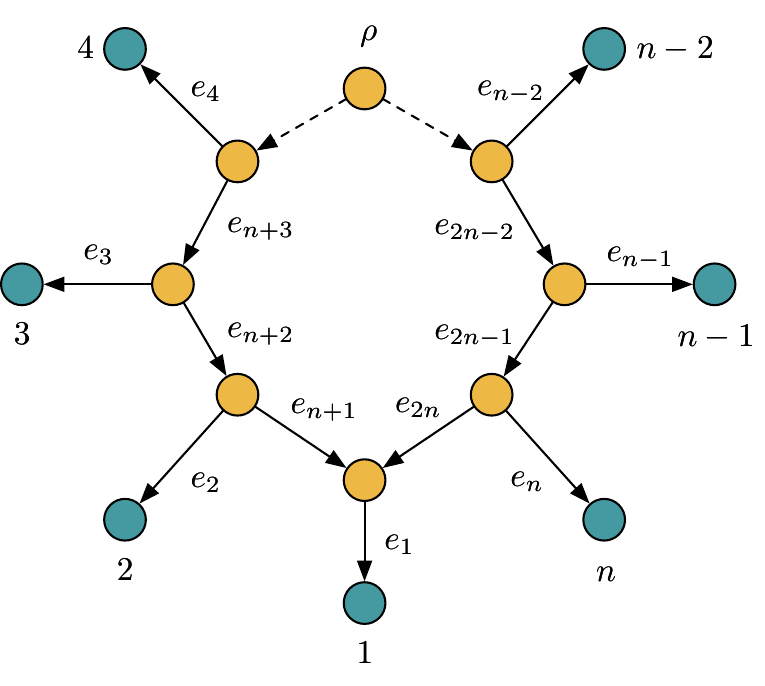}~~~~~~~~
  \includegraphics[scale=0.5]{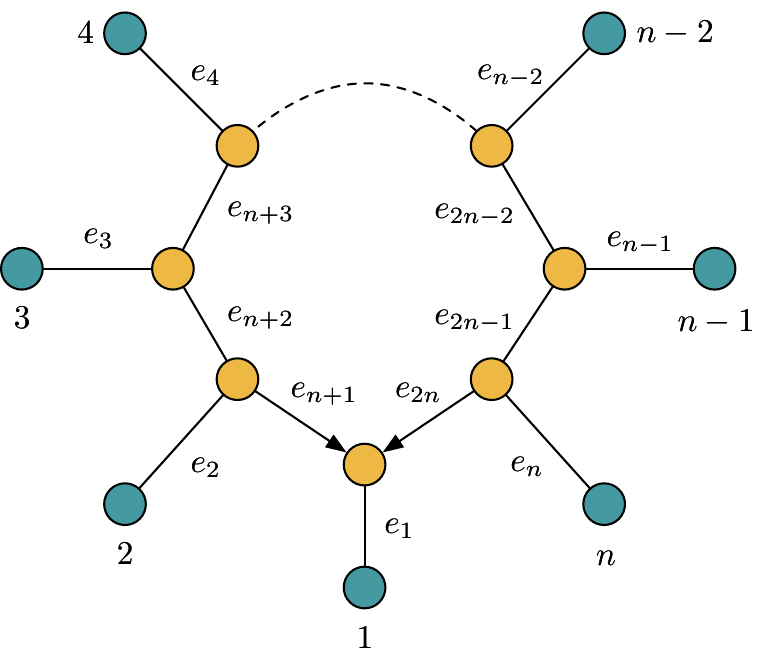}
  \caption{The cycle and sunlet network of $\Gamma$ with edgeset $E=\{e_1,\dots,e_{2n}\}$ on the left and right, respectively.}
  \label{fig:sunlet}
\end{figure}
To simplify our notation we identify the elements of $\Gamma$ with the leaves of the network. Vertices with in-degree one and out-degree two are called \emph{tree vertices} and vertices with in-degree two and out-degree one are called \emph{reticulation vertices}. An edge $(u,v)$ for which $v$ is a reticulation vertex is called a \emph{reticulation edge}. There are some classes of phylogenetic networks that are of particular interest to us:
\begin{definition}
Let $\mathcal{N}$ be a binary phylogenetic network of $\Gamma$. Then, $\mathcal{N}$ is called a
\begin{enumerate}
\item \emph{phylogenetic tree} if it contains no reticulation vertices.
\item \emph{cycle network} if it contains exactly one reticulation vertex and each internal vertex has exactly one leaf as its child.
\item \emph{level-$k$ network} if there exists a maximum of $k$ reticulation vertices in each biconnected component of $\mathcal{N}$.
\item \emph{tree-child network} if every internal vertex has at least one child vertex that is either a tree vertex or a leaf.
\end{enumerate}
\end{definition}
An example of a cycle network is shown in Figure~\ref{fig:sunlet} on the left. Notice that level-1 networks are a subclass of tree-child networks~\cite{Kong22}. Moreover, observe that any level-1 phylogenetic network can be constructed by glueing together cycle networks and phylogenetic trees. 

This article focuses on time-reversible models. This means, we are able to identify reticulation edges but not the location of the root in a phylogenetic network $\mathcal{N}$. Hence, we will mainly be interested in the \emph{semi-directed} network we obtain from $\mathcal{N}$ by suppressing the root and removing the orientation of all tree edges.
\begin{definition}
Let $\mathcal{N}$ be a cycle network of $\Gamma$. Then, we call the semi-directed network of $\mathcal{N}$ the \emph{sunlet network} $S(\mathcal{N})$ of $\Gamma$.
\end{definition}
To simplify our notation we assume that the leaves of the sunlet network $S(\mathcal{N})$ of $\Gamma$ are labelled clockwise by increasing elements in $\Gamma$, starting in the child of the reticulation vertex in $\mathcal{N}$. Furthermore, the edgeset of $S(\mathcal{N})$ is ordered by clockwise enumerating edges which are incident to leaves first and the remaining edges second as shown on the right in Figure~\ref{fig:sunlet}.

\subsection{Algebraic varieties of network-based Markov models}
In this section we equip the edges of a phylogenetic network with transition matrices that follow a random walk Markov model. This allows us to study phylogenetic networks under a model of DNA sequence evolution for a finite set of input data. Subsequently, we consider the algebraic varieties which arise naturally from the study of these models. For an introduction to this algebraic approach from a phylogenetics perspective we refer the reader to~\cite{PS05}.

Given a tree-child phylogenetic network $\mathcal{N}=(V,E)$ of $\Gamma$ and an integer $\kappa\geq 2$, associate a random variable $X_v$ with state space $S_{\kappa}=\{s_1,\dots,s_\kappa\}$ to each vertex $v\in V$. Furthermore, associate a $\kappa\times\kappa$ transition matrix $M^e$ with each edge $e=(u,v)\in E$ such that $$M^e_{ij}=\,\text{Pr}[X_v=j\,|\,X_u=i].$$ In addition, associate a root distribution $\pi$ to the root $\rho$ of $\mathcal{N}$. Let $\mathcal{A}$ denote the set of all possible assignments $A_0:\Gamma\to S_\kappa$ and, for $A_0\in\mathcal{A}$, let $\mathcal{A}(\mathcal{N},A_0)$ denote the set of all possible assignments $A:V\to S_\kappa$ such that $A(v)=A_0(v)$ for all $v\in\Gamma$. The assignments $A_0$ are commonly referred to as characters in the phylogenetics literature. However, in this article we reserve this description for the characters of finite abelian groups. Finally, for the number of reticulation vertices $k$, let $\{\mathcal{T}_1,\dots,\mathcal{T}_{2^k}\}$ denote the set of pairwise distinct phylogenetic trees we obtain from $\mathcal{N}$ by deleting exactly one edge in each pair of adjacent reticulation edges and removing edge subdivisions/vertices with in-degree and out-degree one. Then, for $j\in\{1,\dots,2^k\}$, the joint distribution on the states of the leaves of $T_j$ is given by the probabilities
\begin{align*}
p_{j}(A_0)&:=\text{Pr}\left[(X_v)_{v\in\Gamma}=(A_0(v))_{v\in\Gamma}\right]\\
&\,=\sum_{A\in\mathcal{A}\left(\mathcal{T}_j,A_0\right)}\pi_{A(\rho)}\prod_{e=(u,v)\in E}M^e_{A(u)A(v)}~&~&\forall\,A_0\in\mathcal{A}.
\end{align*}
For the calculation of marginal distributions
\begin{align*}
P_{\mathcal{T}_j}&=\left(p_{j}(A_0)\right)_{A_0\in\mathcal{A}}~&~&\forall\,j\in\{1,\dots,2^k\}
\end{align*}
we assume that transition matrices $M^e$ follow a random walk $\kappa$-state Markov model on the edges $E$. This means, the probabilities $M_{s_is_j}^e$ only depend on the difference of states $s_i$ and $s_j$ for all $e\in E$, $s_i,s_j\in S_{\kappa}$. Hence, we call the model defining distribution $P_{\mathcal{T}_j}$ the \emph{phylogenetic Markov model} on $\mathcal{T}_j$. Let $\Theta_{\mathcal{N}}$ denote the space of transition matrices $M^e$ and root distributions $\pi$, called the \emph{stochastic parameter space}, and let $\Delta_{d-1}=\left\{p\in\mathbb{R}_{\geq 0}^n\,:\,\sum_{i=1}^np_i=1\right\}$ denote the $d$-th dimensional probability simplex. Then, for edges $E=\{e_1,\dots,e_m\}$, we can view the phylogenetic Markov model on $\mathcal{T}_j$ as the image of the polynomial map
\begin{align*}
\phi_{\mathcal{T}_j}:~~~~~~~~~~\Theta_{\mathcal{N}}~~~~~~~~~~&\to\Delta_{\kappa^n-1}\\
(\pi,M^{e_1},\dots,M^{e_m})&\mapsto P_{\mathcal{T}_j}.
\end{align*}
Now, for $j\in\{1,\dots,2^k\}$, consider that we obtain $\mathcal{T}_j$ from $\mathcal{N}$ by independently per reticulation vertex deleting edges $e_i\in E$ with probability $\lambda_i$ and removing edge subdivisions. If $e_i$ is not a reticulation edge, then $\lambda_i=0$. If $(u_k,v_k)$ and $(u_l,v_l)$ are a pair of adjacent reticulation edges, then $v_k=v_l$ because $\mathcal{N}$ is tree-child and therefore $\lambda_k=1-\lambda_l$. Let $\sigma(\mathcal{T}_j)\in\{0,1\}^m$ denote the indicator variable for the edges from $E$ present in $\mathcal{T}_j$. Then, analogously to the maps $\phi_{\mathcal{T}_j}$, we define the phylogenetic Markov model on $\mathcal{N}$ as the image of the polynomial map
\begin{align*}
\phi_{\mathcal{N}}:~~~~~~~~~~\Theta_{\mathcal{N}}~~~~~~~~~~&\to\Delta_{\kappa^n-1}\\
(\pi,M^{e_1},\dots,M^{e_m})&\mapsto \sum_{j=1}^{2^k}\left(\prod_{i=1}^m\lambda_i^{1-\sigma(\mathcal{T}_j)_i}(1-\lambda_i)^{\sigma(\mathcal{T}_j)_i}\right)P_{\mathcal{T}_j}.
\end{align*}
The image of $\phi_{\mathcal{N}}$ is a marginal distribution 
\begin{align*}
P_{\mathcal{N}}=\left(p(A_0)\right)_{A_0\in\mathcal{A}}
\end{align*}
on the leaves of $\mathcal{N}$. Hence, we call a polynomial in the indeterminates $p(A_0)$, $A_0\in\mathcal{A}$, which evaluates to zero for parameters from $\Theta_{\mathcal{N}}$ a \emph{phylogenetic invariant} of the phylogenetic Markov model on $\mathcal{N}$.

Let $\mathbb{C}[p_{i_1\dots i_n}]$ denote the polynomial ring over the indeterminates $p(A_0)$, $A_0\in\mathcal{A}$ with $i_k=A_0(k)$, $k\in\Gamma$. For $j\in\{1,\dots,2^k\}$, it is easy to see that the set of phylogenetic invariants $I_{\mathcal{T}_j}$ of the phylogenetic Markov model on $\mathcal{T}_j$ is a prime ideal in $\mathbb{C}[p_{i_1\dots i_n}]$. Indeed, for $f,g\in\mathbb{C}[p_{i_1\dots i_n}]$ such that $fg\in I_{\mathcal{T}_j}$, we know that $fg$ evaluates to zero for parameters from $\Theta_{\mathcal{N}}$, i.e., $f\in I_{\mathcal{T}_j}$ or $g\in I_{\mathcal{T}_j}$. If both do, then $fg=f+g$, which contradicts the fact that $f$ and $g$ are polynomials. We call the set of phylogenetic invariants $I_{\mathcal{N}}$ of the phylogenetic Markov model on $\mathcal{N}$ the \emph{phylogenetic ideal} of $\mathcal{N}$. Since $I_{\mathcal{N}}$ is the sum of prime ideals, it is not prime. To simplify the analysis, we do not take the structure of the stochastic parameter space $\Theta_{\mathcal{N}}$ into account when studying the phylogenetic ideal. Hence, we extend $\phi_{\mathcal{N}}$ to be a complex polynomial map. We call the variety
\begin{align*}
V_{\mathcal{N}}=\{p\in\mathbb{C}^{\kappa^n}\,:\,f(p)=0~\forall\,f\in I_{\mathcal{N}}\}
\end{align*}
the \emph{phylogenetic variety} of $\mathcal{N}$. Observe that $V_{\mathcal{N}}$ is the Zariski closure of the image of $\phi_{\mathcal{N}}$. Assume for the rest of the article that the phylogenetic Markov model is a random walk 4-state Markov model.

\subsection{Fourier analysis of the phylogenetic variety}
In this subsection we look into the structure of the random walk 4-state Markov model to parameterize the phylogenetic variety of tree-child phylogenetic networks $\mathcal{N}$. We assume that this model is time-reversible, i.e., the location of the root under the phylogenetic Markov model cannot be identified. This means, the root distribution is uniform and the phylogenetic ideal of $\mathcal{N}$ and the ideal of $S(\mathcal{N})$ coincide. 

First, consider the set $\mathbb{G}=\mathbb{Z}_2\oplus\mathbb{Z}_2=\{(0,0),(0,1),(1,0),(1,1)\}$ with a group operation defined by coordinate-wise addition modulo 2. The resulting group $(\mathbb{G},+)$ is called the \emph{Klein four-group}. Then, there exists an isomorphism between the state space $S_4=\{s_1,s_2,s_3,s_4\}$ of our Markov model and $\mathbb{G}$:
\begin{align*}
s_1\leftrightarrow(0,0),~~s_2\leftrightarrow(0,1),~~s_3\leftrightarrow(1,0),~~s_4\leftrightarrow(1,1).
\end{align*}
Recall that $S_4=\Sigma =\{A,G,C,T\}$ is the common identification in the phylogenetics literature. Since the transition matrices $M^e$ follow a random walk, we can apply this isomorphism to state that there exist functions $f^e:\mathbb{G}\to\mathbb{R}$ such that $M_{gh}^e=f^e(g-h)$ for all $e\in E$, $g,h\in\mathbb{G}$. A phylogenetic Markov model which is identified in this way with a group is called a \emph{group-based} model in the phylogenetics literature. Hence, for $j\in\{1,\dots,2^k\}$, we can rewrite the elements of the marginal distribution $P_{\mathcal{T}_j}$ as follows:
\begin{align}\label{p::group}
p_{j,g_1\cdots g_n}&=\sum_{g_v\in\mathbb{G},v\in V\setminus\Gamma}\pi(g_\rho)\prod_{u\,:\,e=(u,v)\in E(\mathcal{T}_j)}f^e(g_u-g_v)~&~&\forall\,g\in\mathbb{G}^n.
\end{align}
The dependence of the transition probabilities on differences of group elements is reminiscent of the convolution operation in Fourier analysis. Indeed, as Sturmfels and Sullivant~\cite{SS05} showed, a Fourier transform of the probabilities $p_{j,g_1\cdots g_n}$ will turn the prime ideals $I_{\mathcal{T}_j}$, $j\in\{1,\dots,2^k\}$, into toric ideals, i.e., ideals which are generated by the differences of monomials. This will help us in the calculation of the phylogenetic variety~$V_{\mathcal{N}}$. To this end, we recall some more definition for finite abelian groups, specifically applied to the Klein four-group: the set $\mathbb{S}^1=\{z\in\mathbb{C}\,:\,|z|=1\}$ with the group operation defined by complex multiplication is called the \emph{circle group}. Then, any map $\chi :\mathbb{G}\to\mathbb{S}^1$ with $\chi(g+h)=\chi(g)\chi(h)$ for all $g,h\in\mathbb{F}$ is called a \emph{character} of $\mathbb{G}$. Let $\hat{\mathbb{G}}$ denote the set of characters of $\mathbb{G}$. Then, $\hat{\mathbb{G}}$ with the group operation defined by the pointwise multiplication of functions is called the \emph{dual group} of $\mathbb{G}$. Let $f:\mathbb{G}\to\mathbb{C}$ be a function. Then, the function $\hat{f}:\hat{\mathbb{G}}\to\mathbb{C}$ defined by
\begin{align*}
\hat{f}(\chi)=\sum_{g\in\mathbb{G}}\chi(g)f(g)
\end{align*}
is called the \emph{Fourier transform} of $f$. 
Before we consider a Fourier parameterization of $P_{\mathcal{N}}$, we provide some intuition for our formulas by illustrating the Fourier analysis done by Evans and Speed~\cite{ES93}: let $\mathcal{T}_j$ be the tree consisting of only a root, $n$ leaves and edges between the root and leaves. Then, for $g\in\mathbb{G}^n$, $p$-coordinates~\eqref{p::group} can be written as
\begin{align*}
p_{g}&=\sum_{h\in\mathbb{G}}\pi(h)\prod_{i=1}^nf^{(\rho,i)}(h-g_i).
\end{align*}
Define function $\mu:\mathbb{G}^n\to\mathbb{R}$ by $\mu(g)=\pi(g_1)$ if $g_1=\dots =g_n$ and $\mu(g)=0$ otherwise, to obtain
\begin{align*}
p_{g}&=\sum_{h\in\mathbb{G}^n}\mu(h)\prod_{i=1}^nf^{(\rho,i)}(h-g_i).
\end{align*}
Now, $p_g$ can be seen as the convolution of functions $\mu$ and $\prod_{i=1}^nf^{\rho,i}$ on the group $\mathbb{G}^n$. Then, applying identities from Fourier analysis and the isomorphism between the Klein four-group and its dual group we arrive at Fourier coordinates
\begin{align*}
q_g&=\hat{\mu}\left(\sum_{i=1}^ng_i\right)\prod_{i=1}^n\widehat{f^{(\rho,i)}}(g_i)~&~&\forall\,g\in\mathbb{G}^n.
\end{align*}
We call the various Fourier transforms on the righthand side the \emph{Fourier parameters} of $\mathcal{T}_j$ and denote them by $a_g^0=\hat{\mu}(g)$ and $a_g^i=\widehat{f^{(\rho,i)}}(g)$ for suitable $g\in\mathbb{G}$. Since we assume that the root distribution $\pi$ is uniform, the Fourier transform $\hat{\mu}$ vanishes if not evaluated at the identity. Therefore, $q_g=0$ for all $g\in\mathbb{G}^n$ with $\sum_{i=1}^ng_i=0$. Similar arguments and definitions can be made for phylogenetic networks. Then, the general parameterization of the sunlet network $S(\mathcal{N})$ (using the enumeration from Figure~\ref{fig:sunlet})~\cite{CHM24} is given by
\begin{align*}
\psi_{\mathcal{N}}:\mathbb{C}\left[q_g\,:\,g\in\mathbb{G}^n,~\sum_{i=1}^ng_i=0\right]&\to\mathbb{C}\left[a_g^i\,:\,g\in\mathbb{G},~i\in\{1,\dots,2n\}\right]\\
q_g&\mapsto\prod_{j=1}^na_{g_j}^j\left(\prod_{j=1}^{n-1}a_{\sum_{l=1}^jg_l}^{n+j}+\prod_{j=2}^na_{\sum_{l=2}^jg_l}^{n+j}\right)
\end{align*}
The kernel of $\psi_{\mathcal{N}}$ is the phylogenetic ideal $I_{\mathcal{N}}$ in Fourier coordinates.

\section{Quadratic invariants for sunlet networks}\label{sec3}
In this section we calculate all quadratic invariants for sunlet networks under the random walk 4-state Markov model. Cummings et al.~\cite{CHM24} proved that the phylogenetic ideal of a level-1 phylogenetic network can be constructed by taking \emph{toric fiber products}~\cite{Sull07} of phylogenetic ideals of trees and cycle/sunlet networks which the level-1 network decomposes into. Since Sturmfels and Sullivant~\cite{SS05} provide a Gr\"obner basis for the phylogenetic invariants of trees, we focus our attention here on sunlet networks, i.e., we assume $\mathcal{N}$ is a cycle network.

First, we generalize the notion of a glove that Cummings et al.~\cite{CHM24} introduced for the binary Markov model to partition the set of quadratic phylogenetic invariants in $I_{\mathcal{N}}$:
\begin{definition}
Let $F_1,F_2\subseteq\Gamma$, $b^1\in(\mathbb{Z}/2\mathbb{Z})^{F_1}$, $b^2\in(\mathbb{Z}/2\mathbb{Z})^{F_2}$, $b^3\in(\mathbb{Z}/2\mathbb{Z})^{\Gamma\setminus(F_1\cup F_{2})}$. Denote $F=\{F_1,F_2\}$ and $b=\{b^1,b^2,b^3\}$. Then, we call the complex vector space
\begin{align*}
\mathcal{G}(n,F,b)=\,\text{span}(q_Gq_H\,&:\,G,H\in(\mathbb{Z}/2\mathbb{Z})^{2\times n},~\sum_{j=1}^nG_{ij}=\sum_{j=1}^nH_{ij}=0~\forall\,i\in\{1,2\},\\
&~~~G_{ij}=H_{ij}=b_j^i~~~\forall\,i\in\{1,2\},j\in F_{i},\\
&~~~G_{ij}+H_{ij}=1~~~~\forall\,i\in\{1,2\},j\in\Gamma\setminus F_{i},\\
&~~~G_{1j}+G_{2j}=b_j^{3}~~\,\forall\,j\in\Gamma\setminus(F_1\cup F_{2}))
\end{align*}
the \emph{2-glove} of $F$ and $b$.
\end{definition}
Observe that every quadratic phylogenetic invariant is contained in $I_{\mathcal{N}}\cap\mathcal{G}(n,F,b)$ for some valid choice of $F$ and $b$. However, we will see that not every 2-glove contains contains a phylogenetic invariant of $\mathcal{N}$. For $F_3\subseteq\Gamma$, we write $\mathcal{G}(n-|F_3|,F,b)$ to denote the 2-glove we obtain from $\mathcal{G}(n,F,b)$ by removing columns of $G,H\in(\mathbb{Z}/2\mathbb{Z})^{2\times n}$ indexed by $F_3$ and restricting $F$ and $b$ accordingly. We assume throughout the rest of the article that all 2-gloves satisfy $F_1\cap F_2 =\emptyset$ if not stated otherwise. This assumption does not restrict $I_{\mathcal{N}}\cap\mathcal{G}(n,F,b)$ as the following proposition shows:

\begin{proposition}\label{prop::Gdim}
Let $F_1,F_2\subseteq\Gamma$, $b^1\in(\mathbb{Z}/2\mathbb{Z})^{F_1}$, $b^2\in(\mathbb{Z}/2\mathbb{Z})^{F_2}$, $b^3\in(\mathbb{Z}/2\mathbb{Z})^{\Gamma\setminus(F_1\cup F_{2})}$. Then, for $F_3=F_1\cap F_2$ with $|F_3|\geq n-1$, we have $I_{\mathcal{N}}\cap\mathcal{G}(n,F,b)\cong I_{\mathcal{N}}\cap\mathcal{G}(n-|F_3|,F,b)$.
\end{proposition}
\begin{proof}
First, consider $q_Gq_H\in\mathcal{G}(n-|F_3|,F,b)$. Let $G'\in(\mathbb{Z}/2\mathbb{Z})^{2\times n}$ be defined by
\begin{align*}
G_{ij}'&=\begin{cases}
G_{ij} &\text{if }j\in\Gamma\setminus F_3,\\
b_j^i &\text{otherwise}
\end{cases}~&~&\forall\,i\in\{1,2\}.
\end{align*}
Define $H'\in(\mathbb{Z}/2\mathbb{Z})^{2\times n}$ analogously. If $\sum_{j\in F_3}b_j^i=0$ for $i\in\{1,2\}$, then $q_{G'}q_{H'}\in\mathcal{G}(n,F,b)$. Otherwise, there exists $i\in\{1,2\}$ such that $\sum_{j\in F_3}b_j^i=1$. In this case, add one to $G_{ij}'$ for $j\in\Gamma\setminus F_3$ maximum.
\end{proof}
Now, we can make some further observations that will help us to characterize the phylogenetic invariants in each 2-glove:
\begin{figure}[t]
  \centering
  \includegraphics[scale=0.5]{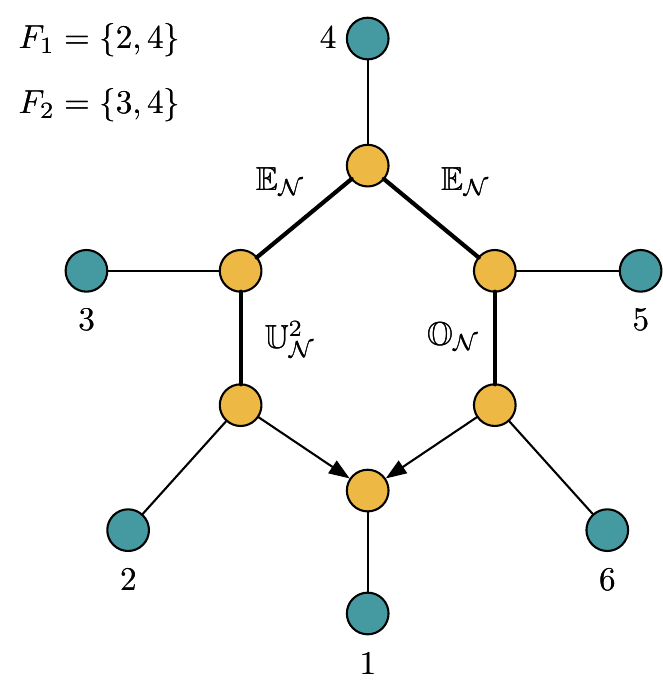}
  \caption{The sunlet network of $\Gamma=\{1,\dots,6\}$. The partition of the bold edges into sets $\mathbb{E}_{\mathcal{N}}$, $\mathbb{O}_{\mathcal{N}}$ and $\mathbb{U}_{\mathcal{N}}$ is shown for $F_1=\{2,4\}$ and $F_2=\{3,4\}$.}
  \label{fig:sunlet2}
\end{figure}
\begin{proposition}\label{prop::dim}
Let $F_1,F_2\subseteq\Gamma$ with $F_1\cap F_2=\emptyset$, $b^1\in(\mathbb{Z}/2\mathbb{Z})^{F_1}$, $b^2\in(\mathbb{Z}/2\mathbb{Z})^{F_2}$, $b^3\in(\mathbb{Z}/2\mathbb{Z})^{\Gamma\setminus(F_1\cup F_{2})}$.
\begin{enumerate}
\item If $|\Gamma\setminus F_1|$ and $|\Gamma\setminus F_2|$ are even, then dim$(\mathcal{G}(n,F,b))=2^{n-2}$.
\item If $|\Gamma\setminus F_1|$ or $|\Gamma\setminus F_2|$ is odd, then $I_{\mathcal{N}}\cap\mathcal{G}(n,F,b)=\emptyset$.
\end{enumerate}
\end{proposition}
\begin{proof}
First, for $q_Gq_H\in\mathcal{G}(n,F,b)$ and $i\in\{1,2\}$, observe that
\begin{align*}
\sum_{j\in\Gamma}G_{ij}&=\sum_{j\in\Gamma\setminus F_i}G_{ij}+\sum_{j\in F_i}b_j^i=0,\\
\sum_{j\in\Gamma}H_{ij}&=\sum_{j\in\Gamma\setminus F_i}G_{ij}+\sum_{j\in\Gamma\setminus F_i}1+\sum_{j\in F_i}b_j^i=0
\end{align*}
only if $|\Gamma\setminus F_1|$ and $|\Gamma\setminus F_2|$ are even. Otherwise, $\mathcal{G}(n,F,b)=\{0\}$. Next, for $q_Gq_H\in\mathcal{G}(n,F,b)$, the choice of $G\in(\mathbb{Z}/2\mathbb{Z})^{2\times n}$ is constrained by $2+|F_1|+|F_2|+n-|F_1\cup F_2|=2+n$ equations because $F_1\cap F_2=\emptyset$. Hence,
\begin{align*}
\text{dim}(\mathcal{G}(n,F,b))=2^{2n-(2+n)}=2^{n-2}.
\end{align*}
\end{proof}

Before we consider examples of how to decide whether a polynomial is in $I_{\mathcal{N}}\cap\mathcal{G}(n,F,b)$ or not we derive two technical results for a series of definitions that align with the methodology of Cummings et al.~\cite{CHM24} and which are useful throughout the rest of this article. First, for the sunlet network $S(\mathcal{N})$, we define a partition of edges $e_{n+1},\dots,e_{2n-1}$ from Figure~\ref{fig:sunlet} into the four sets
\begin{align*}
\mathbb{E}_{\mathcal{N}}&=\left\{e_{n+j}\in E(\mathcal{N})\,:\,\left|\{1,\dots,j\}\setminus F_1\right|~\text{and}~\left|\{1,\dots,j\}\setminus F_2\right|~\text{even},~j\in\{2,\dots,n-1\}\right\},\\
\mathbb{O}_{\mathcal{N}}&=\left\{e_{n+j}\in E(\mathcal{N})\,:\,\left|\{1,\dots,j\}\setminus F_1\right|~\text{and}~\left|\{1,\dots,j\}\setminus F_2\right|~\text{odd},~j\in\{2,\dots,n-1\}\right\},\\
\mathbb{U}_{\mathcal{N}}^1&=\left\{e_{n+j}\in E(\mathcal{N})\,:\,\left|\{1,\dots,j\}\setminus F_1\right|~\text{even and}~\left|\{1,\dots,j\}\setminus F_2\right|~\text{odd},~j\in\{2,\dots,n-1\}\right\},\\
\mathbb{U}_{\mathcal{N}}^2&=\left\{e_{n+j}\in E(\mathcal{N})\,:\,\left|\{1,\dots,j\}\setminus F_1\right|~\text{odd and}~\left|\{1,\dots,j\}\setminus F_2\right|~\text{even},~j\in\{2,\dots,n-1\}\right\}.
\end{align*}
We write $\mathbb{U}_{\mathcal{N}}=\mathbb{U}_{\mathcal{N}}^1\cup\mathbb{U}_{\mathcal{N}}^2$ as a reoccurring shorthand notation later in this section. Figure~\ref{fig:sunlet2} shows an example of this partition for $n=6$. Observe that the partition $(\mathbb{E}_{\mathcal{N}},\mathbb{O}_{\mathcal{N}},\mathbb{U}_{\mathcal{N}})$ of edges in $\mathcal{N}$ corresponds to a partition of Fourier parameters in $\psi_{\mathcal{N}}$:
\begin{lemma}\label{lem:colors}
Let $F_1,F_2\subseteq\Gamma$, $b^1\in(\mathbb{Z}/2\mathbb{Z})^{F_1}$, $b^2\in(\mathbb{Z}/2\mathbb{Z})^{F_2}$, $b^3\in(\mathbb{Z}/2\mathbb{Z})^{\Gamma\setminus(F_1\cup F_{2})}$. Let $q_Gq_H\in\mathcal{G}(n,F,b)$, $i\in\{1,2\}$, $k\in\Gamma$ and $F_3\subseteq\Gamma$.
\begin{enumerate}
\item If $k\notin F_3$ or $k\in F_3$ and $G_{ik}=0$ if $k\in F_i$, $G_{ik}=1$ otherwise, then $\left|F_3\setminus F_i\right|$ is even if and only if $\sum_{j\in F_3}G_{ij}=\sum_{j\in F_3\setminus\{k\}}H_{ij}$.
\item If $k\in F_3$ and $G_{ik}=0$ if $k\in\Gamma\setminus F_i$, $G_{ik}=1$ otherwise, then $\left|F_3\setminus F_i\right|$ is odd if and only if $\sum_{j\in F_3}G_{ij}=\sum_{j\in F_3\setminus\{k\}}H_{ij}$.
\end{enumerate}
\end{lemma}
A proof of Lemma~\ref{lem:colors} is given in Appendix~\ref{appendix}. Now, we can relate our partition of Fourier parameters to the image of $\psi_{\mathcal{N}}$. Subsequently, this will allow us to study a binary encoding of non-redundant structures in the image of $\psi_{\mathcal{N}}$ which will aid our construction of $I_{\mathcal{N}}\cap\mathcal{G}(n,F,b)$.

For $G\in(\mathbb{Z}/2\mathbb{Z})^{2\times n}$, $i\in\{1,2\}$, $j\in\Gamma$, let $G_{i\cdot}$ and $G_{\cdot j}$ denote the $i$-th row and $j$-th column of $G$, respectively. Let $L(n,F,b)$ denote the set of matrices $G\in(\mathbb{Z}/2\mathbb{Z})^{2\times n}$ such that $q_Gq_H\in\mathcal{G}(n,F,b)$ and the vector $(G_{1\cdot},G_{2\cdot})$ is lexicographically smaller than $(H_{1\cdot},H_{2\cdot})$. Then, for $G\in L(n,F,b)$, a subset $S$ of internal edges in $\mathcal{N}$ and $r,s\in\{1,2\}$, let matrices $$c_{r,s}^{=}(S,G),~c_{r,s}^{\neq}(S,G),~c_{r,s}^{\sim}(S,G)\in(\mathbb{Z}/2\mathbb{Z})^{2\times|S|}$$ be defined by
\begin{align*}
c_{r,s}^=(S,G)&=\left[\begin{array}{c}\sum_{l=r}^jG_{1l}\\ \sum_{l=s}^jH_{2l}\end{array}\right]_{e_{n+j}\in S},~&~~c_{r,s}^{\neq}(S,G)&=\left[\begin{array}{c}\sum_{l=r}^jG_{1l}+\sum_{l=s}^jH_{2l}\\ 0\end{array}\right]_{e_{n+j}\in S}
\end{align*}
and, for $e_{n+j}\in S$,
\begin{align*}
c_{r,s}^{\sim}(S,G)_{\cdot j}&=\begin{cases}
\left(\begin{array}{c}\sum_{l=r}^jG_{1l}+\sum_{l=s}^jH_{1l}\\ \sum_{l=r}^jG_{2l}\end{array}\right) &\text{if }\sum\limits_{l=r}^j\left(G_{1l}+G_{2l}\right)=0,\\
\left(\begin{array}{c}\sum_{l=r}^jG_{1l}+\sum_{l=s}^jH_{1l}\\ \sum_{l=s}^jH_{2l}\end{array}\right) &\text{otherwise,}
\end{cases}
\end{align*}
The purpose of these matrices is to encode the relationship between vectors $v(G,j)=\sum_{l=r}^jG_{\cdot l}$ and $v(H,j)=\sum_{l=s}^jH_{\cdot l}$ with respect to edges $e_{n+j}\in S$: if $v(G,j)=v(H,j)$ or $v(G,j)\neq v(H,j)$, then the pair $v(G,j)$ and $v(H,j)$ can be uniquely identified by $c_{r,s}^{=}(S,G)$ or $c_{r,s}^{\neq}(S,G)$, respectively. If $v(G,j)_1=v(H,j)_1$ and $v(G,j)_2\neq v(H,j)_2$, then $c_{r,s}^{\sim}(S,G)=(0,v(G,j)_1)$, and, if $v(G,j)_1\neq v(H,j)_1$ and $v(G,j)_2=v(H,j)_2$, then $c_{r,s}^{\sim}(S,G)=(1,v(G,j)_2)$. Hence, the remaining pairs of vectors $v(G,j)$ and $v(H,j)$ can be uniquely identified by $c_{r,s}^{\sim}(S,G)$. An alternative representation of matrix $c_{r,s}^{\sim}(S,G)$ which will be useful later in this section is given by
\begin{align*}
c_{r,s}^{\sim}(S,G)_{\cdot j}&=\begin{cases}
\left(\begin{array}{c}\sum_{l=r}^jG_{1l}+\sum_{l=s}^jH_{1l}\\ \sum_{l=s}^jH_{1l}\end{array}\right) &\text{if }\sum\limits_{l=s}^j\left(H_{1l}+H_{2l}\right)=0,\\
\left(\begin{array}{c}\sum_{l=r}^jG_{1l}+\sum_{l=s}^jH_{1l}\\ \sum_{l=r}^jG_{1l}\end{array}\right) &\text{otherwise,}
\end{cases}~&~&\forall\,e_{n+j}\in S.
\end{align*}
We will make use of these definitions with the following linear maps: for $r,s\in\{1,2\}$ and pairwise disjoint sets $S_1,S_2,S_3\subseteq\{e_2,\dots,e_{n-1}\}$ with $|S_1\cup S_2\cup S_3|=n-2$, let
\begin{align*}
M_{2r+s-2}^{F,b}(S_1,S_2,S_3):\mathcal{G}(n,F,b)&\to\mathbb{C}^{(\mathbb{Z}/2\mathbb{Z})^{2\times(n-2)}}\\
q_Gq_H&\mapsto e_{c_{r,s}^{=}(S_1,G)\oplus c_{r,s}^{\neq}(S_2,G)\oplus c_{r,s}^{\sim}(S_3,G)}
\end{align*}
where $e_i$ denotes the standard basis vector in $\mathbb{C}^m$. For $i\in\{1,\dots,4\}$, we write $M_i^{F,b}(\mathbb{E}_{\mathcal{N}})$, $M_i^{F,b}(\mathbb{O}_{\mathcal{N}})$, $M_i^{F,b}(\mathbb{U}_{\mathcal{N}}^1)$ and $M_i^{F,b}(\mathbb{U}_{\mathcal{N}}^2)$ as shorthand notation for 
\begin{align*}
&M_i^{F,b}(\mathbb{E}_{\mathcal{N}},\mathbb{O}_{\mathcal{N}},\mathbb{U}_{\mathcal{N}}),~~M_i^{F,b}(\mathbb{O}_{\mathcal{N}},\mathbb{E}_{\mathcal{N}},\mathbb{U}_{\mathcal{N}}),~~M_i^{F,b}(\mathbb{U}_{\mathcal{N}}^1,\mathbb{U}_{\mathcal{N}}^2,\mathbb{E}_{\mathcal{N}}\cup\mathbb{O}_{\mathcal{N}})\\
&\text{and}~~M_i^{F,b}(\mathbb{U}_{\mathcal{N}}^2,\mathbb{U}_{\mathcal{N}}^1,\mathbb{E}_{\mathcal{N}}\cup\mathbb{O}_{\mathcal{N}}),
\end{align*}
respectively. Furthermore, for $S_1,\dots,S_4\in\{\mathbb{E}_{\mathcal{N}},\mathbb{O}_{\mathcal{N}},\mathbb{U}_{\mathcal{N}}^1,\mathbb{U}_{\mathcal{N}}^2\}$, we define
\begin{align*}
\text{ker}(S_1,S_2,S_3,S_4):=\bigcap_{j=1}^4\text{ker}\left(M_j^{F,b}\left(S_j\right)\right)
\end{align*}
to improve the readability of the following proposition whose proof is given in Appendix~\ref{appendix}:

\begin{proposition}\label{prop::IG}
Let $S(\mathcal{N})$ be the sunlet network of $\Gamma$ and let $\mathcal{T}$ denote the tree we obtain from $\mathcal{N}$ by deleting leaf $1$ and its neighbour in $\mathcal{N}$. Let $F_1,F_2\subseteq\Gamma$,  $b^1\in(\mathbb{Z}/2\mathbb{Z})^{F_1}$, $b^2\in(\mathbb{Z}/2\mathbb{Z})^{F_2}$, $b^3\in(\mathbb{Z}/2\mathbb{Z})^{\Gamma\setminus(F_1\cup F_{2})}$. Then,
\begin{align*}
I_{\mathcal{N}}\cap\mathcal{G}(n,F,b)&=\begin{cases}I_{\mathcal{T}}\cap\mathcal{G}(n,F,b) &\text{if }1\in F_1\cap F_2,~(b_1^1,b_1^2)=(0,0),\\
\mathrm{ker}\left(M_1^{F,b}(\mathbb{E}_{\mathcal{N}})\right)\cap\,\mathrm{ker}\left(M_2^{F,b}(\mathbb{U}_{\mathcal{N}}^1)\right) &\text{if }1\in F_1\cap F_2,~(b_1^1,b_1^2)=(0,1),\\
&\text{or }1\in F_1,~1\notin F_2,~b_1^1=0,\\
\mathrm{ker}\left(M_1^{F,b}(\mathbb{E}_{\mathcal{N}})\right)\cap\,\mathrm{ker}\left(M_2^{F,b}(\mathbb{U}_{\mathcal{N}}^2)\right) &\text{if }1\in F_1\cap F_2,~(b_1^1,b_1^2)=(1,0),\\
&\text{or }1\notin F_1,~1\in F_2,~b_1^2=0,\\
\mathrm{ker}\left(M_1^{F,b}(\mathbb{E}_{\mathcal{N}})\right)\cap\,\mathrm{ker}\left(M_2^{F,b}(\mathbb{O}_{\mathcal{N}})\right) &\text{if }1\in F_1\cap F_2,~(b_1^1,b_1^2)=(1,1),\\
&\text{or }1\notin F_1\cup F_2,~b_1^3=0,\\
\mathrm{ker}\,(\mathbb{E}_{\mathcal{N}},\mathbb{U}_{\mathcal{N}}^2,\mathbb{U}_{\mathcal{N}}^1,\mathbb{O}_{\mathcal{N}}) &\text{if }1\notin F_1\cup F_2,~b_1^3=1,\\
\mathrm{ker}\,(\mathbb{E}_{\mathcal{N}},\mathbb{O}_{\mathcal{N}},\mathbb{U}_{\mathcal{N}}^1,\mathbb{U}_{\mathcal{N}}^2) &\text{otherwise.}
\end{cases}
\end{align*}
\end{proposition}
Now, we can apply Proposition~\ref{prop::IG} to test the membership of polynomials in $I_{\mathcal{N}}\cap\mathcal{G}(n,F,b)$. For example, consider $n=4$ and polynomials
\begin{align*}
f_1&=q_{\left(\begin{array}{cccc} 0 & 0 & 0 & 0\\ 0 & 0 & 1 & 1\end{array}\right)}q_{\left(\begin{array}{cccc} 1 & 1 & 1 & 1\\ 1 & 1 & 0 & 0\end{array}\right)}-q_{\left(\begin{array}{cccc} 0 & 0 & 1 & 1\\ 0 & 0 & 0 & 0\end{array}\right)}q_{\left(\begin{array}{cccc} 1 & 1 & 0 & 0\\ 1 & 1 & 1 & 1\end{array}\right)},\\
f_2&=q_{\left(\begin{array}{cccc} 0 & 0 & 0 & 0\\ 0 & 1 & 0 & 1\end{array}\right)}q_{\left(\begin{array}{cccc} 1 & 0 & 0 & 1\\ 1 & 0 & 1 & 0\end{array}\right)}-q_{\left(\begin{array}{cccc} 0 & 0 & 0 & 0\\ 0 & 0 & 1 & 1\end{array}\right)}q_{\left(\begin{array}{cccc} 1 & 0 & 0 & 1\\ 1 & 1 & 0 & 0\end{array}\right)},
\end{align*}
in the $q$-coordinates on the sunlet network. Then, $$f_1\in\mathcal{G}(4,\emptyset,\{(0,0,1,1)\})~~~\text{and}~~~f_2\in\mathcal{G}(4,\{\{2,3\},\emptyset\},\{(0,0),(0,1)\}).$$ Since $f_1$ and $f_2$ are in distinct 2-gloves, $f_1+f_2$ can not be in a 2-glove and is therefore not a phylogenetic invariant of $\mathcal{N}$. For $F=\emptyset$, we have
\begin{align*}
\mathbb{E}_{\mathcal{N}}=\{e_{4+2}\},~\mathbb{O}_{\mathcal{N}}=\{e_{4+3}\}~~\text{and}~~\mathbb{U}_{\mathcal{N}}=\emptyset,
\end{align*}
and for $F=\{\{2,3\},\emptyset\}$, we have
\begin{align*}
\mathbb{E}_{\mathcal{N}}=\emptyset,~\mathbb{O}_{\mathcal{N}}=\{e_{4+3}\},~\mathbb{U}_{\mathcal{N}}^1=\emptyset~~\text{and}~~\mathbb{U}_{\mathcal{N}}^2=\{e_{4+2}\}.
\end{align*}
From Proposition~\ref{prop::IG} we know that $f_i$, $i\in\{1,2\}$, is a phylogenetic invariant of $\mathcal{N}$ only if 
\begin{align*}
f_i\in\,\text{ker}\left(M_1^{\mathcal{F},b}(\mathbb{E}_{\mathcal{N}})\right)\cap\,\text{ker}\left(M_2^{\mathcal{F},b}(\mathbb{O}_{\mathcal{N}})\right)
\end{align*}
because in both cases $1\notin F_1\cup F_2$ and $b_1^3=0$. Then, for $f_1$ we get
\begin{align*}
c_{1,1}^{=}\left(\mathbb{E}_{\mathcal{N}},\left(\begin{array}{cccc} 0 & 0 & 0 & 0\\ 0 & 0 & 1 & 1\end{array}\right)\right)&=\left(\begin{array}{c}0\\ 0\end{array}\right)=c_{1,1}^{=}\left(\mathbb{E}_{\mathcal{N}},\left(\begin{array}{cccc} 0 & 0 & 1 & 1\\ 0 & 0 & 0 & 0\end{array}\right)\right)\\
c_{1,1}^{\neq}\left(\mathbb{O}_{\mathcal{N}},\left(\begin{array}{cccc} 0 & 0 & 0 & 0\\ 0 & 0 & 1 & 1\end{array}\right)\right)&=\left(\begin{array}{c}0\\ 0\end{array}\right)=c_{1,1}^{\neq}\left(\mathbb{O}_{\mathcal{N}},\left(\begin{array}{cccc} 0 & 0 & 1 & 1\\ 0 & 0 & 0 & 0\end{array}\right)\right)\\
c_{1,2}^{=}\left(\mathbb{O}_{\mathcal{N}},\left(\begin{array}{cccc} 0 & 0 & 0 & 0\\ 0 & 0 & 1 & 1\end{array}\right)\right)&=\left(\begin{array}{c}0\\ 1\end{array}\right)\neq \left(\begin{array}{c}1\\ 0\end{array}\right)=c_{1,2}^{=}\left(\mathbb{O}_{\mathcal{N}},\left(\begin{array}{cccc} 0 & 0 & 1 & 1\\ 0 & 0 & 0 & 0\end{array}\right)\right)\\
c_{1,2}^{\neq}\left(\mathbb{E}_{\mathcal{N}},\left(\begin{array}{cccc} 0 & 0 & 0 & 0\\ 0 & 0 & 1 & 1\end{array}\right)\right)&=\left(\begin{array}{c}0\\ 0\end{array}\right)=c_{1,2}^{\neq}\left(\mathbb{E}_{\mathcal{N}},\left(\begin{array}{cccc} 0 & 0 & 1 & 1\\ 0 & 0 & 0 & 0\end{array}\right)\right)
\end{align*}
Hence,
\begin{align*}
\left(M_1^{\mathcal{F},b}(\mathbb{E}_{\mathcal{N}})\right)(f_1)&=e_{(0,0)\oplus(0,0)}-e_{(0,0)\oplus(0,0)}=0,\\
\left(M_2^{\mathcal{F},b}(\mathbb{O}_{\mathcal{N}})\right)(f_1)&=e_{(0,1)\oplus (0,0)}-e_{(1,0)\oplus(0,0)}\neq 0.
\end{align*}
Thus, $f_1$ is not a phylogenetic invariant of $\mathcal{N}$. Contrary, for $f_2$ we get
\begin{align*}
c_{1,1}^{\neq}\left(\mathbb{O}_{\mathcal{N}},\left(\begin{array}{cccc} 0 & 0 & 0 & 0\\ 0 & 1 & 0 & 1\end{array}\right)\right)&=\left(\begin{array}{c}0\\ 0\end{array}\right)=c_{1,1}^{\neq}\left(\mathbb{O}_{\mathcal{N}},\left(\begin{array}{cccc} 0 & 0 & 0 & 0\\ 0 & 0 & 1 & 1\end{array}\right)\right)\\
c_{1,1}^{\sim}\left(\mathbb{U}_{\mathcal{N}},\left(\begin{array}{cccc} 0 & 0 & 0 & 0\\ 0 & 1 & 0 & 1\end{array}\right)\right)&=\left(\begin{array}{c}1\\ 0\end{array}\right)=c_{1,1}^{\sim}\left(\mathbb{U}_{\mathcal{N}},\left(\begin{array}{cccc} 0 & 0 & 0 & 0\\ 0 & 0 & 1 & 1\end{array}\right)\right)\\
c_{1,2}^{=}\left(\mathbb{O}_{\mathcal{N}},\left(\begin{array}{cccc} 0 & 0 & 0 & 0\\ 0 & 1 & 0 & 1\end{array}\right)\right)&=\left(\begin{array}{c}0\\ 1\end{array}\right)=c_{1,2}^{=}\left(\mathbb{O}_{\mathcal{N}},\left(\begin{array}{cccc} 0 & 0 & 0 & 0\\ 0 & 0 & 1 & 1\end{array}\right)\right)\\
c_{1,2}^{\sim}\left(\mathbb{U}_{\mathcal{N}},\left(\begin{array}{cccc} 0 & 0 & 0 & 0\\ 0 & 1 & 0 & 1\end{array}\right)\right)&=\left(\begin{array}{c}0\\ 1\end{array}\right)=c_{1,2}^{\sim}\left(\mathbb{U}_{\mathcal{N}},\left(\begin{array}{cccc} 0 & 0 & 0 & 0\\ 0 & 0 & 1 & 1\end{array}\right)\right)
\end{align*}
This means,
\begin{align*}
\left(M_1^{\mathcal{F},b}(\mathbb{E}_{\mathcal{N}})\right)(f_2)&=e_{(0,0)\oplus(1,0)}-e_{(0,0)\oplus(1,0)}=0,\\
\left(M_2^{\mathcal{F},b}(\mathbb{O}_{\mathcal{N}})\right)(f_2)&=e_{(0,1)\oplus (0,1)}-e_{(0,1)\oplus(0,1)}=0.
\end{align*}
Thus, $f_2$ is a phylogenetic invariant of $\mathcal{N}$. 

We will show that considering specific differences of binomials like $f_2$ are sufficient to generate all quadratic phylogenetic invariants of $\mathcal{N}$. Therefore, we need to formalize the observations in our example. To this end, we define a first ingredient:

\begin{definition}\label{def:basicMon}
Let $F_1,F_2\subseteq\Gamma$ with $F_1\cap F_2=\emptyset$, $b^1\in(\mathbb{Z}/2\mathbb{Z})^{F_1}$, $b^2\in(\mathbb{Z}/2\mathbb{Z})^{F_2}$, $b^3\in(\mathbb{Z}/2\mathbb{Z})^{\Gamma\setminus(F_1\cup F_{2})}$. Let $c\in(\mathbb{Z}/2\mathbb{Z})^{n}$ with $c_1=c_n=0$ and let $G(b,c),H(b,c)\in(\mathbb{Z}/2\mathbb{Z})^{2\times n}$ be defined by $G(b,c)_{\cdot j}=G^1(b,c)_{\cdot j}+G^2(b,c)_{\cdot j}$ for all $j\in\{1,\dots,n-1\}$ with
\begin{align*}
G^1(b,c)_{\cdot j}&=\begin{cases}
\left(b_j^1,c_j\right) &\text{if }j\in F_1,\\
\left(c_j,b_j^2\right) &\text{if }j\in F_2,\\
\left(c_j,c_j+b_j^3\right) &\text{otherwise},
\end{cases}~\\
G^2(b,c)_{\cdot j}&=\begin{cases}
(0,c_{j-1}) &\text{if }e_{n+j-1}\in\mathbb{U}_{\mathcal{N}}\text{ and }(j-1),j\in F_1,\\
(0,c_{j}) &\text{if }e_{n+j}\in\mathbb{U}_{\mathcal{N}}\text{ and }j\in F_1,\,(j+1)\notin F_1,\\
(c_{j-1},0) &\text{if }e_{n+j-1}\in\mathbb{U}_{\mathcal{N}}\text{ and }(j-1),j\in F_2,\\
(c_{j},0) &\text{if }e_{n+j}\in\mathbb{U}_{\mathcal{N}}\text{ and }j\in F_2,\,(j+1)\notin F_2,\\
(0,0) &\text{otherwise,}
\end{cases}~\\
G(b,c)_{\cdot n}&=\sum_{j=1}^{n-1}G(b,c)_{\cdot j}
\end{align*}
and $q_{G(b,c)}q_{H(b,c)}\in\mathcal{G}(n,F,b)$. Then, we call $q_{G(b,c)}q_{H(b,c)}$ the \emph{basic monomial} of $\mathcal{G}(n,F,b)$ and~$c$.
\end{definition}

\noindent We continue our example for binomial $f_2$ to illustrate Definition~\ref{def:basicMon}. For $f_2\in\mathcal{G}(4,F,b)$, recall that $F=\{\{2,3\},\emptyset\}$ and $b=\{(0,0),(0,1)\}$. Then, for $c=(0,1,1,0)$,
\begin{align*}
G^1(b,c)&=\left(\begin{array}{ccc} c_1 & b_2^1 & b_3^1\\ c_1+b_1^3 & c_2 & c_3\end{array}\right)=\left(\begin{array}{ccc} 0 & 0 & 0 \\ 0 & 1 & 1 \end{array}\right)\\
G^2(b,c)&=\left(\begin{array}{ccc} 0 & 0 & 0 \\ 0 & 0 & c_2 \end{array}\right)=\left(\begin{array}{ccc} 0 & 0 & 0 \\ 0 & 0 & 1 \end{array}\right)
\end{align*}
Furthermore, for $d=(0,0,1,0)$,
\begin{align*}
G^1(b,d)&=\left(\begin{array}{ccc} 0 & 0 & 0 \\ 0 & 0 & 1 \end{array}\right),~~~G^2(b,c)=\left(\begin{array}{ccc} 0 & 0 & 0 \\ 0 & 0 & 0 \end{array}\right).
\end{align*}
This means, $f_2$ is the difference of basic monomials:
\begin{align*}
f_2=q_{G(b,c)}q_{H(b,c)}-q_{G(b,d)}q_{H(b,d)}.
\end{align*}

The second ingredient to construct all quadratic phylogenetic invariants of $\mathcal{N}$ is a strengthening of Proposition~\ref{prop::dim}:
\begin{proposition}\label{prop::dim2}
Let $n\geq 4$ and let $S(\mathcal{N})$ be the sunlet network of $\Gamma$. Let $F_1,F_2\subseteq\Gamma$ with $F_1\cap F_2=\emptyset$, $b^1\in(\mathbb{Z}/2\mathbb{Z})^{F_1}$, $b^2\in(\mathbb{Z}/2\mathbb{Z})^{F_2}$, $b^3\in(\mathbb{Z}/2\mathbb{Z})^{\Gamma\setminus(F_1\cup F_{2})}$ such that $|\Gamma\setminus F_i|$ is even for $i\in\{1,2\}$. Then, $$I_{\mathcal{N}}\cap\mathcal{G}(n,F,b)\neq\emptyset$$ if and only if $b_1^3=0$, $1\notin F_1\cup F_2$ and
\begin{align*}
\left|\{e_{n+j}\in\mathbb{U}_{\mathcal{N}}\,:\,j\notin F_1~\vee~(j-1)\notin F_1,j\in F_1~\vee~j\notin F_2~\vee~(j-1)\notin F_2,j\in F_2\}\right|>0.
\end{align*}
\end{proposition}
\begin{proof}
First, assume
\begin{align}\label{prop::IG::apply}
I_{\mathcal{N}}\cap\mathcal{G}(n,F,b)=\,\text{ker}\left(M_1^{F,b}(\mathbb{E}_{\mathcal{N}})\right)\cap\,\text{ker}\left(M_2^{F,b}(S_1)\right)
\end{align}
for some $S_1\in\{\mathbb{O}_{\mathcal{N}},\mathbb{U}_{\mathcal{N}}^1,\mathbb{U}_{\mathcal{N}}^2\}$. 
Then, for $q_Gq_H\in\mathcal{G}(n,F,b)$, we have
\begin{align}
c_{1,1}^=(\mathbb{E}_{\mathcal{N}},G)&=\left[\begin{array}{c}\sum_{l=1}^jG_{1l}\\ \sum_{l=1}^jG_{2l}\end{array}\right]_{e_{n+j}\in\mathbb{E}_{\mathcal{N}}},\label{M1::1}\\
c_{1,1}^{\neq}(\mathbb{O}_{\mathcal{N}},G)&=\left[\begin{array}{c}\sum_{l=1}^j\left(G_{1l}+H_{2l}\right)\\ 0\end{array}\right]_{e_{n+j}\in\mathbb{O}_{\mathcal{N}}}\nonumber\\
&=\left[\begin{array}{c}\sum_{l=1}^j(G_{1l}+G_{2l})+|\{1,\dots,j\}\setminus F_2|\\ 0\end{array}\right]_{e_{n+j}\in\mathbb{O}_{\mathcal{N}}},\label{M1::2}
\end{align}
and, for $e_{n+j}\in\mathbb{U}_{\mathcal{N}}$,
\begin{align}
c_{1,1}^{\sim}(\mathbb{U}_{\mathcal{N}},G)_{\cdot j}&=\begin{cases}
\left(\begin{array}{c}|\{1,\dots,j\}\setminus F_1|\\ \sum_{l=1}^jG_{1l}\end{array}\right) &\text{if }\sum\limits_{l=1}^j\left(G_{1l}+G_{2l}\right)=0,\\
\left(\begin{array}{c}|\{1,\dots,j\}\setminus F_1|\\ \sum_{l=1}^jG_{1l}+|\{1,\dots,j\}\setminus F_2|+1\end{array}\right) &\text{otherwise.}
\end{cases}\label{M1::3}
\end{align}
and equivalently,
\begin{align}
c_{1,1}^{\sim}(\mathbb{U}_{\mathcal{N}},G)_{\cdot j}&=\begin{cases}
\left(\begin{array}{c}|\{1,\dots,j\}\setminus F_1|\\ \sum_{l=1}^jH_{2l}\end{array}\right) &\text{if }\sum\limits_{l=1}^j\left(H_{1l}+H_{2l}\right)=0,\\
\left(\begin{array}{c}|\{1,\dots,j\}\setminus F_1|\\ \sum_{l=1}^jH_{2l}+|\{1,\dots,j\}\setminus F_1|+1\end{array}\right) &\text{otherwise.}
\end{cases}\label{M1::4}
\end{align}
Similarly, for 
\begin{align*}
S_2&=\begin{cases}
\mathbb{E}_{\mathcal{N}} &\text{if }S_1=\mathbb{O}_{\mathcal{N}},\\
\mathbb{U}_{\mathcal{N}}^2 &\text{if }S_1=\mathbb{U}_{\mathcal{N}}^1,\\
\mathbb{U}_{\mathcal{N}}^1 &\text{if }S_1=\mathbb{U}_{\mathcal{N}}^2,
\end{cases}~~~&~~~S_3&=\begin{cases}
\mathbb{U}_{\mathcal{N}} &\text{if }S_1=\mathbb{O}_{\mathcal{N}},\\
\mathbb{E}_{\mathcal{N}}\cup\mathbb{O}_{\mathcal{N}} &\text{otherwise},
\end{cases}
\end{align*}
we have
\begin{align}
c_{1,2}^=(S_1,G)&=\left[\begin{array}{c}\sum_{l=1}^jG_{1l}\\ \sum_{l=2}^jG_{2l}\end{array}\right]_{e_{n+j}\in S_1},\label{M2::1}\\
c_{1,2}^{\neq}(S_2,G)
&=\left[\begin{array}{c}H_{21}+\sum_{l=1}^j(G_{1l}+G_{2l})+|\{1,\dots,j\}\setminus F_2|\\ 0\end{array}\right]_{e_{n+j}\in S_2},\label{M2::2}
\end{align}
and, for $e_{n+j}\in S_3$,
\begin{align}
c_{1,2}^{\sim}(S_3,G)_{\cdot j}&=\begin{cases}
\left(\begin{array}{c}G_{11}+|\{2,\dots,j\}\setminus F_1|\\ \sum_{l=1}^jG_{1l}\end{array}\right) &\text{if }\sum\limits_{l=1}^j\left(G_{1l}+G_{2l}\right)=0,\\
\left(\begin{array}{c}G_{11}+|\{2,\dots,j\}\setminus F_1|\\ \sum_{l=2}^jG_{1l}+|\{2,\dots,j\}\setminus F_2|+1\end{array}\right) &\text{otherwise.}
\end{cases}\label{M2::3}
\end{align}
Since $F_1\cap F_2=\emptyset$, the vector $G_{\cdot j}$ is defined by exactly one free choice in its entries for all $G\in L(n,F,b)$, $j\in\{2,\dots,n-1\}$. Hence, from Equations~\eqref{M1::1} to~\eqref{M1::3} we know that matrices $c_{1,1}^{=}(\mathbb{E}_{\mathcal{N}},G)$, $c_{1,1}^{\neq}(\mathbb{O}_{\mathcal{N}},G)$, $c_{1,1}^{\sim}(\mathbb{U}_{\mathcal{N}},G)$ have at least
\begin{align*}
d_{\mathbb{E}}=|\mathbb{E}_{\mathcal{N}}|,~d_{\mathbb{O}}=|\{e_{n+j}\in\mathbb{O}_{\mathcal{N}}\,:\,j\in F_1\cup F_2\}|,~d_{\mathbb{U}}=|\{e_{n+j}\in\mathbb{U}_{\mathcal{N}}\,:\,j\in\Gamma\setminus(F_1\cup F_2)\}|
\end{align*}
degrees of freedom, respectively. Here we used the fact that $G_{1j}+G_{2j}$ is not constant only if $j\in F_1\cup F_2$, and that $G_{1j}$ is a linear function of $G_{2j}$ in the definition of $c_{1,1}^{\sim}(\mathbb{U}_{\mathcal{N}},G)_{\cdot j}$. Furthermore, for $e_{n+j}\in\mathbb{U}_{\mathcal{N}}$ with $(j-1)\notin F_2$ and $j\in F_2$ we observe that column~\eqref{M1::3} is independent of the choice of $G_{\cdot j}$. Analogously, for $(j-1)\notin F_1$ and $j\in F_1$, column~\eqref{M1::4} is independent of the choice of $G_{\cdot j}$. This means, we conclude that matrix $c_{1,1}^{\sim}(\mathbb{U}_{\mathcal{N}},G)$ has
\begin{align*}
\hat{d}_{\mathbb{U}}=d_{\mathbb{U}}+|\{e_{n+j}\in\mathbb{U}_{\mathcal{N}}\,:\,(j-1)\notin F_1,j\in F_1~\vee~(j-1)\notin F_2,j\in F_2\}|
\end{align*}
degrees of freedom. Therefore,
\begin{align*}
\text{rank}\left(M_1^{F,b}(\mathbb{E}_{\mathcal{N}})\right)=2^{d_{\mathbb{E}}+d_{\mathbb{O}}+\hat{d}_{\mathbb{U}}}.
\end{align*}
Analogously, matrices $c_{1,2}^{=}(S_1,G)$, $c_{1,2}^{\neq}(S_2,G)$ and $c_{1,2}^{\sim}(S_3,G)$ have 
\begin{align*}
d_{1}&=|S_1|,~~d_{2}=|\{e_{n+j}\in S_2\,:\,j\in F_1\cup F_2\}|\\
\text{and}~d_{3}&=|\{e_{n+j}\in S_3\,:\,j\in\Gamma\setminus(F_1\cup F_2)~\vee~(j-1)\notin F_1,j\in F_1~\vee~(j-1)\notin F_2,j\in F_2\}|
\end{align*}
degrees of freedom, respectively, i.e.,
\begin{align*}
\text{rank}\left(M_2^{F,b}(S_1)\right)=2^{d_{1}+d_{2}+d_{3}}.
\end{align*}
From these observations we can conclude that
\begin{align*}
\text{rank}\left(M_1^{F,b}(\mathbb{E}_{\mathcal{N}})\oplus M_2^{F,b}(S_1)\right)&=\begin{cases}
2^{|\mathbb{E}_{\mathcal{N}}|+|\mathbb{O}_{\mathcal{N}}|+\hat{d}_{\mathbb{U}}} &\text{if }S_1=\mathbb{O}_{\mathcal{N}},\\
2^{|\mathbb{E}_{\mathcal{N}}|+|\mathbb{O}_{\mathcal{N}}|+|\mathbb{U}_{\mathcal{N}}|} &\text{otherwise.}
\end{cases}
\end{align*}
This means, by Equation~\eqref{prop::IG::apply}, Proposition~\ref{prop::dim} and the rank-nullity theorem, we arrive at
\begin{align*}
\text{dim}\left(\text{ker}\left(M_1^{F,b}(\mathbb{E}_{\mathcal{N}})\oplus M_2^{F,b}(S_1)\right)\right)&=\,\text{dim}(\mathcal{G}(n,F,b))-\,\text{rank}\left(M_1^{F,b}(\mathbb{E}_{\mathcal{N}})\oplus M_2^{F,b}(S_1)\right)\\
&=\begin{cases}
2^{n-2}-2^{|\mathbb{E}_{\mathcal{N}}|+|\mathbb{O}_{\mathcal{N}}|+\hat{d}_{\mathbb{U}}} &\text{if }S_1=\mathbb{O}_{\mathcal{N}},\\
0 &\text{otherwise.}
\end{cases}
\end{align*}
Analogously, $\text{dim}\left(\text{ker}\left(\mathbb{E}_{\mathcal{N}},\mathbb{U}_{\mathcal{N}}^2,\mathbb{U}_{\mathcal{N}}^1,\mathbb{O}_{\mathcal{N}}\right)\right)=0$ and $\text{dim}\left(\text{ker}\left(\mathbb{E}_{\mathcal{N}},\mathbb{O}_{\mathcal{N}},\mathbb{U}_{\mathcal{N}}^1,\mathbb{U}_{\mathcal{N}}^2\right)\right)=0$. Thus, our claim follows from Proposition~\ref{prop::IG}.
\end{proof}

Now, using basic monomials and Propositions~\ref{prop::IG} and~\ref{prop::dim2} we are ready to construct a basis for all non-empty partitioned ideals $I_{\mathcal{N}}\cap\mathcal{G}(n,F,b)$:
\begin{theorem}\label{thm::basis}
Let $n\geq 4$ and let $S(\mathcal{N})$ be the sunlet network of $\Gamma$. Let $F_1,F_2\subseteq\Gamma$ with $F_1\cap F_2=\emptyset$, $|\Gamma\setminus F_i|$ even for $i\in\{1,2\}$ and $1\notin F_1\cup F_2$, $b^1\in(\mathbb{Z}/2\mathbb{Z})^{F_1}$, $b^2\in(\mathbb{Z}/2\mathbb{Z})^{F_2}$, $b^3\in(\mathbb{Z}/2\mathbb{Z})^{\Gamma\setminus(F_1\cup F_{2})}$ with $b_1^3=0$. Let
\begin{align*}
R=\{j\,:\,e_{n+j}\in\mathbb{E}_{\mathcal{N}}\cup\mathbb{O}_{\mathcal{N}}&~\vee~e_{n+j}\in\mathbb{U}_{\mathcal{N}},\,j\in\Gamma\setminus(F_1\cup F_2)\\
&~\vee~e_{n+j}\in\mathbb{U}_{\mathcal{N}},\,(j-1)\notin F_1,j\in F_1\\
&~\vee~e_{n+j}\in\mathbb{U}_{\mathcal{N}},\,(j-1)\notin F_2,j\in F_2\}\subseteq\Gamma.
\end{align*}
and let $c,d\in(\mathbb{Z}/2\mathbb{Z})^{n}$ with $c_1=c_n=d_1=d_n=0$, $c_j=d_j$ for all $j\in R$ and $c_j=0$ for all $j\in\Gamma\setminus R$. Let $f_{b,c,d}$ denote the difference of the basic monomials of $\mathcal{G}(n,F,b)$ and $c$ and $d$, respectively. Then, $f_{b,c,d}$ is a phylogenetic invariants of $\mathcal{N}$. Furthermore, let $Z$ denote the set of tuples $(c,d)$ defining $f_{b,c,d}$. If
\begin{align*}
\mathcal{B}(n,F,b)&=\left\{f_{b,c,d}\,:\,(c,d)\in Z,~c\neq d\right\}
\end{align*}
is non-empty, then $\mathcal{B}(n,F,b)$ is a basis for $I_\mathcal{N}\cap\mathcal{G}(n,F,b)$.
\end{theorem}
\begin{proof}
Since $b_1^3=0$ and $1\notin F_1\cup F_2$, we know from Proposition~\ref{prop::IG} that
\begin{align*}
I_{\mathcal{N}}\cap\mathcal{G}(n,F,b)=\,\text{ker}\left(M_1^{F,b}(\mathbb{E}_{\mathcal{N}})\right)\cap\,\text{ker}\left(M_2^{F,b}(\mathbb{O}_{\mathcal{N}})\right)
\end{align*}
and we know from the proof of Proposition~\ref{prop::dim2} that
\begin{align*}
\text{dim}\left(I_{\mathcal{N}}\cap\mathcal{G}(n,F,b)\right)=2^{n-2}-2^{|R|}.
\end{align*}
First, consider matrices~\eqref{M1::1} to~\eqref{M1::3}. Define $C\in(\mathbb{Z}/2\mathbb{Z})^{2\times n}$ as follows: for $e_{n+j}\in\mathbb{E}_{\mathcal{N}}$ and $x\in\{c,d\}$, 
\begin{align*}
C_{\cdot j}:=&\,c_{1,1}^{=}(\mathbb{E}_{\mathcal{N}},G(b,x))_{\cdot j}\\
=&\sum_{l=1}^jG(b,x)_{\cdot l}\\
=&\left(\begin{array}{c} 0\\ b_1^3 \end{array}\right)+\sum_{l=2,\,l\notin F_1\cup F_2}^j\left(\begin{array}{c}x_l\\ x_l+b_l^3\end{array}\right)+\sum_{l=2,\,l\in F_1}^j\left(\begin{array}{c}b_l^1\\ x_l\end{array}\right)+\sum_{l=2,\,l\in F_2}^j\left(\begin{array}{c}x_l\\ b_l^2\end{array}\right)\\
&~~+\sum_{l=2,\,e_{n+l-1}\in\mathbb{U}_{\mathcal{N}},\,(l-1),l\in F_1}^j\left(\begin{array}{c}0\\ x_{l-1}\end{array}\right)+\sum_{l=2,\,e_{n+l-1}\in\mathbb{U}_{\mathcal{N}},\,(l-1),l\in F_2}^j\left(\begin{array}{c}x_{l-1}\\ 0\end{array}\right)\\
&~~+\sum_{l=2,\,e_{n+l}\in\mathbb{U}_{\mathcal{N}},\,l\in F_1,\,(l+1)\notin F_1}^j\left(\begin{array}{c}0\\ x_{l}\end{array}\right)+\sum_{l=2,\,e_{n+l}\in\mathbb{U}_{\mathcal{N}},\,l\in F_2,\,(l+1)\notin F_2}^j\left(\begin{array}{c}x_{l}\\ 0\end{array}\right)
\end{align*}
The column vector $C_{\cdot j}$ is well-defined because $d_l=c_l$ for all $e_{n+l}\in\mathbb{E}_{\mathcal{N}}\cup\mathbb{O}_{\mathcal{N}}$ and $e_{n+l}\in\mathbb{U}_{\mathcal{N}}$ with $l\notin F_1\cup F_2$. Indeed, $e_{n+l}\in\mathbb{U}_{\mathcal{N}}$ with $l\in F_1$ implies that either edge $e_{n+l+1}$ contributes $\left(b_{l+1}^1,x_{l+1}+x_l\right)$ or edge $e_{n+l}$ contributes $\left(b_l^1,0\right)$. Hence, $x_l$ does not appear in $C_{\cdot j}$. The same argument works for $F_2$ instead of $F_1$. Analogously,
for $e_{n+j}\in\mathbb{O}_{\mathcal{N}}$ and $e_{n+j}\in\mathbb{U}_{\mathcal{N}}$,
\begin{align*}
C_{\cdot j}:=&\,c_{1,1}^{\neq}(\mathbb{O}_{\mathcal{N}},G(b,x))_{\cdot j}~~~\text{and}~~~C_{\cdot j}:=c_{1,1}^{\sim}(\mathbb{U}_{\mathcal{N}},G(b,x))_{\cdot j}
\end{align*}
are well-defined because the feasibility of our construction these columns for $C$ does not depend on the edge partition $(\mathbb{E}_{\mathcal{N}},\mathbb{O}_{\mathcal{N}},\mathbb{U}_{\mathcal{N}})$. Therefore, we conclude
\begin{align*}
\left(M_1^{F,b}(\mathbb{E}_{\mathcal{N}})\right)(f_{b,c,d})&=e_{c_{1,1}^{=}(\mathbb{E}_{\mathcal{N}},G(b,c))\oplus c_{1,1}^{\neq}(\mathbb{O}_{\mathcal{N}},G(b,c))\oplus c_{1,1}^{\sim}(\mathbb{U}_{\mathcal{N}},G(b,c))}\\
&~~~-e_{c_{1,1}^{=}(\mathbb{E}_{\mathcal{N}},G(b,d))\oplus c_{1,1}^{\neq}(\mathbb{O}_{\mathcal{N}},G(b,d))\oplus c_{1,1}^{\sim}(\mathbb{U}_{\mathcal{N}},G(b,d))}\\
&=e_{C}-e_C=0.
\end{align*}
We define $D\in(\mathbb{Z}/2\mathbb{Z})^{2\times n}$ similar as $C$ using matrices~\eqref{M2::1} to~\eqref{M2::3} instead of matrices~\eqref{M1::1} to~\eqref{M1::3} to conclude that $f_{b,c,d}\in\,\text{ker}\left(M_2^{F,b}(\mathbb{O}_{\mathcal{N}})\right)$, too. Thus, we conclude that $f_{b,c,d}$ is a phylogenetic invariant of $\mathcal{N}$. 

Since there are $2^{|R|}$ choices for $c$ and $2^{n-2-|R|}$ choices for $d-c$, we conclude that $\mathcal{B}(n,F,b)$ has cardinality 
\begin{align*}
2^{|R|}\left(2^{n-2-|R|}-1\right)=2^{n-2}-2^{|R|}=\,\text{dim}\left(I_{\mathcal{N}}\cap\mathcal{G}(n,F,b)\right).
\end{align*}
\end{proof}

Notice that there might exist many more bases of $I_{\mathcal{N}}\cap\mathcal{G}(n,F,b)$ different from our construction $\mathcal{B}(n,F,b)$. This question becomes interesting when we want to compute a basis for each partitioned ideal $I_{\mathcal{N}}\cap\mathcal{G}(n,F,b)$ in practice. Indeed, we can find a second class of bases very similar to the construction in Theorem~\ref{thm::basis} that allows us to exploit the symmetry of 2-gloves:

\begin{corollary}\label{cor::basis}
Let $n\geq 4$ and let $S(\mathcal{N})$ be the sunlet network of $\Gamma$. Let $F_1,F_2\subseteq\Gamma$ with $F_1\cap F_2=\emptyset$ and $1\notin F_1\cup F_2$, $b^1\in(\mathbb{Z}/2\mathbb{Z})^{F_1}$, $b^2\in(\mathbb{Z}/2\mathbb{Z})^{F_2}$, $b^3\in(\mathbb{Z}/2\mathbb{Z})^{\Gamma\setminus(F_1\cup F_{2})}$ with $b_1^3=0$. Then, exchanging the first and second row of indices $G(b,c)$ and $H(b,c)$ defining basic monomials $q_{G(b,c)}q_{H(b,c)}$, yields a basis for $I_{\mathcal{N}}\cap\mathcal{G}(n,F,b)$ if it exists.
\end{corollary}
\begin{proof}
From Theorem~\ref{thm::basis} we know that $\mathcal{B}(n,F,b)$ is a basis for $I_{\mathcal{N}}\cap\mathcal{G}(n,F,b)$. Analogously we obtain a basis for $I_{\mathcal{N}}\cap\mathcal{G}(n,F,b)$ by exchanging the first and second row of indices $G(b,c)$ and $H(b,c)$ of basic monomials $q_{G(b,c)}q_{H(b,c)}$ defining $\mathcal{B}(n,F,b)$.
\end{proof}

Corollary~\ref{cor::basis} allows us to repurpose the basis of a 2-glove to obtain a basis for another distinct 2-glove in linear time. Finally, for small positive integers $n$, we can detail Proposition~\ref{prop::dim2} and Theorem~\ref{thm::basis} further:

\begin{corollary}\label{cor::n4empty}
Let $n\leq 4$ be a positive integer and let $S(\mathcal{N})$ be the sunlet network of $\Gamma$. Let $F_1,F_2\subseteq\Gamma$ with $F_1\cap F_2=\emptyset$, $b^1\in(\mathbb{Z}/2\mathbb{Z})^{F_1}$, $b^2\in(\mathbb{Z}/2\mathbb{Z})^{F_2}$, $b^3\in(\mathbb{Z}/2\mathbb{Z})^{\Gamma\setminus(F_1\cup F_{2})}$ such that $|\Gamma\setminus F_i|$ is even for $i\in\{1,2\}$. Then, $I_{\mathcal{N}}\cap\mathcal{G}(n,F,b)=\emptyset$ for $F=\emptyset$, $n\leq 3$ and $n=4$ with $\mathcal{B}(n,F,b)=\emptyset$. Moreover, if $n=4$, $F\neq\emptyset$ and $\mathcal{B}(n,F,b)\neq\emptyset$, then
\begin{align*}
\mathcal{B}(n,F,b)=&\left\{q_{G(b,(0,0,0,0))}q_{H(b,(0,0,0,0))}-q_{G(b,(0,0,1,0))}q_{H(b,(0,0,1,0))},\right.\\
&~\,\left.q_{G(b,(0,1,1,0))}q_{H(b,(0,1,1,0))}-q_{G(b,(0,1,0,0))}q_{H(b,(0,1,0,0))}\right\}
\end{align*}
\end{corollary}
\begin{proof}
For $F=\emptyset$, by definition $|\mathbb{E}_{\mathcal{N}}|=|\mathbb{O}_{\mathcal{N}}|=n/2-1$ and $|\mathbb{U}_{\mathcal{N}}|=0$. Then, $I_{\mathcal{N}}\cap\mathcal{G}(n,F,b)=\emptyset$ follows from Proposition~\ref{prop::dim2}. Otherwise, for $n\leq 3$, Propositions~\ref{prop::dim} and~\ref{prop::dim2} imply $I_{\mathcal{N}}\cap\mathcal{G}(n,F,b)=\emptyset$, too. Therefore, assume $F\neq\emptyset$ and $n=4$. Since $I_{\mathcal{N}}\cap\mathcal{G}(n,F,b)\neq\emptyset$ only if $1\notin F_1\cup F_2$ and we assumed $F_1\cap F_2=\emptyset$, we know from Proposition~\ref{prop::dim} that $|F_i|=2$ and $|F_l|=0$ for $i,l\in\{1,2\}$, $i\neq l$. Without loss of generality $F_1=\{j,k\}$, $j<k$, and let $q_Gq_H\in\mathcal{G}(n,F,b)$. 

First, assume $\mathcal{G}(n,F,b)=\{q_Gq_H\}$. Since $\psi_{\mathcal{N}}$ is a ring homomorphism, $\psi_{\mathcal{N}}(q_Gq_H)=\psi_{\mathcal{N}}(q_G)\psi_{\mathcal{N}}(q_H)$. By definition, $\psi_{\mathcal{N}}(q_G),\psi_{\mathcal{N}}(q_G)\neq 0$ and the co-domain of $\psi_{\mathcal{N}}$ is integral. Thus, $\psi_{\mathcal{N}}(q_G)\psi_{\mathcal{N}}(q_H)\neq 0$. In other words, $q_Gq_H\notin I_{\mathcal{N}}$.

Next, assume $\mathcal{B}(n,F,b)$ is non-empty, i.e., $\text{dim}(\mathcal{G}(n,F,b))>1$. Consider the set $R$ in Theorem~\ref{thm::basis}. By definition, $j\in R$ and $k\in\Gamma\setminus R$. This means, for $(c,d)\in Z$ with $c\neq d$ we have $c_j=d_j$, $c_k=0$ and $d_k=1$. Thus, our claim holds.
%
%
%
\end{proof}

\section{Summary and conclusions}\label{sec4}
The previous section gives us an outline of an algorithm to calculate all quadratic phylogenetic invariants of the sunlet network $S(\mathcal{N})$ for any fixed number of taxa $n$. First, consider only 2-gloves with $|F_i|\leq n-4$ even for $i\in\{1,2\}$. Any other 2-glove of $F$ is empty (see Proposition~\ref{prop::dim}). Then, establish the assumption $F_1\cap F_2=\emptyset$ justified by Proposition~\ref{prop::Gdim}, i.e., project 2-glove $\mathcal{G}(n,F,b)$ onto 2-glove $\mathcal{G}(n-|F_3|,F,b)$. Next, calculate the partition $(\mathbb{E}_{\mathcal{N}'},\mathbb{O}_{\mathcal{N}'},\mathbb{U}_{\mathcal{N}'})$ of the interior edges in $\mathcal{N}'$ where $\mathcal{N}'$ is the cycle network on $n-|F_3|$ taxa. Then, apply Theorem~\ref{thm::basis} to obtain a basis for $I_{\mathcal{N}'}\cap\mathcal{G}(n-|F_3|,F,b)$ when possible, i.e. the set $\mathcal{B}(n,F,b)$ in Theorem~\ref{thm::basis} is not empty. Here, the assumptions $1\notin F_1\cup F_2$ and $b_1^3=0$ of Theorem~\ref{thm::basis} are justified by Proposition~\ref{prop::dim2}. This means, if these assumptions are violated, then there exist no quadratic phylogenetic invariants of $\mathcal{N}'$ in $\mathcal{G}(n-|F_3|,F,b)$ which are not also invariants of the tree $\mathcal{T}$ covered in Proposition~\ref{prop::IG}. Finally, lifting the invariants back up from all 2-gloves $\mathcal{G}(n-|F_3|,F,b)$ to 2-gloves $\mathcal{G}(n,F,b)$ and adding the quadratic phylogenetic invariants of the tree $\mathcal{T}$ which are well-known in the literature~\cite{SS05}, we obtain all quadratic phylogenetic invariants in $I_{\mathcal{N}}$. 

Now, our study of the sunlet network tells us some things about the phylogenetic ideal of level-1 networks. First, recall that by taking the toric fiber product of cycle networks and trees we obtain a class of invariants in the phylogenetic ideal of a level-1 network from our computations for the sunlet network~\cite{Sull07,CHM24}. However, we also know from Corollary~\ref{cor::n4empty} that $I_{\mathcal{N}}$ does not contain any quadratic invariants for $n\leq 3$ that capture more structure than that of single trees nested inside~$\mathcal{N}$. In other words, all invariants in sets which are sufficient to prove identifiability under our model are at least cubic for most networks on level-1 quarnets. Indeed, Corollary~\ref{cor::n4empty} can provide generators of the phylogenetic ideal only for one specific quarnet, the cycle network for $n=4$. This means, we cannot follow the methodology of Gross et al.~\cite{GIJ21} to give an alternative proof of the identifiability of level-1 phylogenetic networks under the phylogenetic Markov model using only linear and quadratic invariants of trees and sunlet networks. However, since Theorem~\ref{thm::basis} is applicable for $n\geq 4$, the presence of sufficiently large cycles in a level-1 phylogenetic network enables the study of identifiability with the clear methodological advantage of not relying on a randomized matroid-separation algorithm~\cite{HS21}.

\subsection{Future research}

A first direction to extend our results is the study of level-2 (or more broadly level-$k$) phylogenetic network. Here, the connection between the notions of clique sums and toric fiber products~\cite{SS08} can lead the investigation of more complex graph substructure than sunlet networks. Since our methodology provided a better structural understanding of level-1 networks, this research can avoid the pitfalls previous studies on level-2 networks have encountered~\cite{Ardi21}. Our insights also motivate the study of higher degree phylogenetic invariants. In particular, studying invariants of the phylogenetic Markov model via Pfaffians~\cite{Long20} could make use of the results in this article to characterize higher degree phylogenetic invariants for some tree-child networks in a similar fashion to the work of Cummings et al.~\cite{CGH23}. Another research direction concerns a further generalization of our analysis to random walk $\kappa$-state Markov models with $\kappa =2^k$ for $k\geq 3$. Keeping powers of two for the number of states ensures that there exists a clear binary encoding for elements of the cyclic group.

\appendix
\section{Proofs}\label{appendix}

\subsection*{Proof of Lemma~\ref{lem:colors}}
Let $F_3\subseteq\Gamma$, $i\in\{1,2\}$ and let $g,h\in(\mathbb{Z}/2\mathbb{Z})^n$ denote the $i$-th row of $G$ and $H$, respectively. First, we prove our claim for $k\notin F_3$. Assume $|F_3\setminus F_i|$ is even. Then, $\sum_{j\in F_3\setminus F_i}(g_{j}+h_{j})=0$. This means, $\sum_{j\in F_3\setminus F_i}g_{j}$ and $\sum_{j\in F_3\setminus F_i}h_{j}$ are both even or odd. Moreover, 
\begin{align*}
\sum_{j\in F_3\cap F_i}(g_{j}+h_{j})=2\sum_{j\in F_3\cap F_i}b_j^i=0.
\end{align*}
Hence, $\sum_{j\in F_3\cap F_i}g_{j}$ and $\sum_{j\in F_3\cap F_i}h_{j}$ are both even or odd. Thus, $\sum_{j\in F_3}g_{j}$ and $\sum_{j\in F_3}h_{j}$ are both even or odd. The converse follows from similar arguments by assuming $|F_3\setminus F_i|$ is odd. Now, consider our claim for $k\in F_3$ and $k\notin F_i$. If $|F_3\setminus F_i|$ is odd, then
\begin{align*}
\sum_{j\in F_3\setminus\{F_i\cup\{k\}\}}(g_j+h_j)=0
\end{align*}
With analogous arguments as before we conclude that $$\sum_{j\in F_3\setminus\{k\}}g_j=\sum_{j\in F_3\setminus\{k\}}h_j.$$ Hence, for $g_k=1$ we have $\sum_{j\in F_3}g_j\neq\sum_{j\in F_3\setminus\{k\}}h_j$ and for $g_k=0$ we have $\sum_{j\in F_3}g_j=\sum_{j\in F_3\setminus\{k\}}h_j$. Otherwise, $|F_3\setminus F_i|$ is even and the converse follows from similar arguments. For $k\in F_3$ and $k\in F_i$ analogous arguments prove our claim.

\subsection*{Proof of Proposition~\ref{prop::IG}}
\begin{description}
\item[Case 1:] $1\in F_1$, $b^1_1=0$ and $1\in F_2$, $b^2_1=0$. Then,
\begin{align*}
\psi_{\mathcal{N}}(q_G)&=\prod_{j=1}^na_{G_{\cdot j}}^j\left(\prod_{j=1}^{n-1}a_{\sum_{l=1}^jG_{\cdot l}}^{n+j}+\prod_{j=2}^na_{\sum_{l=2}^jG_{\cdot l}}^{n+j}\right)\\
&=a_{G_{\cdot 1}}^1\prod_{j=2}^na_{G_{\cdot j}}^j\left(a_{G_{\cdot 1}}^{n+1}\prod_{j=2}^{n-1}a_{\sum_{l=1}^jG_{\cdot l}}^{n+j}+a_{\sum_{l=2}^nG_{\cdot l}}^{2n}\prod_{j=2}^{n-1}a_{\sum_{l=2}^jG_{\cdot l}}^{n+j}\right)\\
&=a_{G_{\cdot 1}}^1\left(a_{G_{\cdot 1}}^{n+1}+a_{\sum_{l=2}^nG_{\cdot l}}^{2n}\right)\prod_{j=2}^na_{G_{\cdot j}}^j\prod_{j=2}^{n-1}a_{\sum_{l=1}^jG_{\cdot l}}^{n+j}
\end{align*}
Let $\psi_{\mathcal{T}}$ denote the parameterization of the undirected tree underlying $\mathcal{T}$ in Fourier coordinates~\cite{SS05}:
\begin{align*}
\psi_{\mathcal{T}}:\mathbb{C}\left[q_g\,:\,g\in\mathbb{G}^n,~\sum_{i=1}^ng_i=0\right]&\to\mathbb{C}\left[a_g^i\,:\,g\in\mathbb{G},~i\in\{1,\dots,2n\}\right]\\
q_g&\mapsto\prod_{j=1}^na_{g_j}^j\prod_{j=2}^{n-1}a_{\sum_{l=2}^jg_l}^{n+j}
\end{align*}

Since $G_{\cdot 1}=(0,0)$, we know that matrix $G'$ obtained from $G$ by deleting the first column lies in the domain of $\psi_{\mathcal{T}}$. Hence,
\begin{align*}
\psi_{\mathcal{N}}(q_G)=a_{(0,0)}^1\left(a_{(0,0)}^{n+1}+a_{(0,0)}^{2n}\right)\psi_{\mathcal{T}}(G').
\end{align*}
This means, $\psi_{\mathcal{N}}$ and $\psi_{\mathcal{T}}$ have the same kernel.
\item[Case 2:] $1\notin F_1$ or $1\notin F_2$. Let $\beta=G_{\cdot 1}$ and $\gamma =H_{\cdot 1}$. Then, $\beta\neq\gamma$ and, for 
\begin{align*}
f=\sum_{G\in L(n, F,b)}c_Gq_Gq_H\in\mathcal{G}(n, F,b)
\end{align*}
we have
\begin{align*}
&\psi_{\mathcal{N}}(f)\\
&=\sum_{G\in L(n, F,b)}c_G\prod_{j=1}^na_{G_{\cdot j}}^j\left(\prod_{j=1}^{n-1}a_{\sum_{l=1}^jG_{\cdot l}}^{n+j}+\prod_{j=2}^na_{\sum_{l=2}^jG_{\cdot l}}^{n+j}\right)\\
&~~~~~~~~~~~~~~~~~\times\prod_{j=1}^na_{H_{\cdot j}}^j\left(\prod_{j=1}^{n-1}a_{\sum_{l=1}^jH_{\cdot l}}^{n+j}+\prod_{j=2}^na_{\sum_{l=2}^jH_{\cdot l}}^{n+j}\right)\\
&=\sum_{G\in L(n, F,b)}c_G\,\left(\prod_{j=1}^na_{G_{\cdot j}}^ja_{H_{\cdot j}}^j\right)\times\\
&~~~~\left(a_{G_{\cdot 1}}^{n+1}a_{H_{\cdot 1}}^{n+1}\prod_{j=2}^{n-1}a_{\sum_{l=1}^jG_{\cdot l}}^{n+j}a_{\sum_{l=1}^jH_{\cdot l}}^{n+j}+a_{G_{\cdot 1}}^{n+1}a_{\sum_{l=2}^{n}H_{\cdot l}}^{2n}\prod_{j=2}^{n-1}a_{\sum_{l=1}^jG_{\cdot l}}^{n+j}a_{\sum_{l=2}^jH_{\cdot l}}^{n+j}\right.\\
&~~~~~~\left.+a_{\sum_{l=2}^nG_{\cdot l}}^{2n}a_{H_{\cdot 1}}^{n+1}\prod_{j=2}^{n-1}a_{\sum_{l=2}^jG_{\cdot l}}^{n+j}a_{\sum_{l=1}^jH_{\cdot l}}^{n+j}+a_{\sum_{l=2}^nG_{\cdot l}}^{2n}a_{\sum_{l=2}^{n}H_{\cdot l}}^{2n}\prod_{j=2}^{n-1}a_{\sum_{l=2}^jG_{\cdot l}}^{n+j}a_{\sum_{l=2}^jH_{\cdot l}}^{n+j}\right)
\end{align*}

For a fixed $G\in L(n, F,b)$, no pair of products can cancel out because $\beta\neq\gamma$. Moreover, for $j\in\Gamma$, the set of vectors $\{G_{\cdot j},H_{\cdot j}\}$ is fully determined. Hence, $c_G':=c_G\prod_{j=1}^na_{G_{\cdot j}}^ja_{H_{\cdot j}}^j$ is a constant that only depends on $G$. Therefore, $\psi_{\mathcal{N}}(f)=0$ if and only if
\begin{align}
\sum_{G\in L(n, F,b)}c_G'\left(\prod_{j=2}^{n-1}a_{\sum_{l=1}^jG_{\cdot l}}^{n+j}a_{\sum_{l=1}^jH_{\cdot l}}^{n+j}\right)&=0,\label{C2::1}\\
\sum_{G\in L(n, F,b)}c_G'\left(\prod_{j=2}^{n-1}a_{\sum_{l=1}^jG_{\cdot l}}^{n+j}a_{\sum_{l=2}^jH_{\cdot l}}^{n+j}\right)&=0,\label{C2::2}\\
\sum_{G\in L(n, F,b)}c_G'\left(\prod_{j=2}^{n-1}a_{\sum_{l=2}^jG_{\cdot l}}^{n+j}a_{\sum_{l=1}^jH_{\cdot l}}^{n+j}\right)&=0.\label{C2::3}\\
\sum_{G\in L(n, F,b)}c_G'\left(\prod_{j=2}^{n-1}a_{\sum_{l=2}^jG_{\cdot l}}^{n+j}a_{\sum_{l=2}^jH_{\cdot l}}^{n+j}\right)&=0,\label{C2::4}
\end{align}
From Lemma~\ref{lem:colors} we know that
\begin{align*}
e_{n+j}\in\mathbb{E}_{\mathcal{N}}~~~\Leftrightarrow~~~\sum_{l=1}^jH_{\cdot l}=\sum_{l=1}^jG_{\cdot l}.
\end{align*}
Hence, Equation~\eqref{C2::1} holds if and only if
\begin{align*}
\sum_{G\in L(n, F,b)}c_G'M_1^{F,b}\left(\mathbb{E}_{\mathcal{N}},\mathbb{O}_{\mathcal{N}},\mathbb{U}_{\mathcal{N}}\right)&=0.
\end{align*}
\begin{description}
\item[Case 2.1:] $1\notin F_1\cup F_2$ and $\beta =(0,0)$. Then, $\gamma=(1,1)$. Hence, Equation~\eqref{C2::4} holds if and only if
\begin{align*}
\sum_{G\in L(n, F,b)}c_G'M_4^{F,b}\left(\mathbb{O}_{\mathcal{N}},\mathbb{E}_{\mathcal{N}},\mathbb{U}_{\mathcal{N}}\right)&=0.
\end{align*}
From Lemma~\ref{lem:colors} we get
\begin{align*}
e_{n+j}\in\mathbb{O}_{\mathcal{N}}~~~&\Leftrightarrow~~~\sum_{l=1}^jG_{\cdot l}=\sum_{l=2}^jH_{\cdot l},~&e_{n+j}\in\mathbb{E}_{\mathcal{N}}~~~&\Leftrightarrow~~~\sum_{l=1}^jH_{\cdot l}=\sum_{l=2}^jG_{\cdot l}.
\end{align*}
Therefore, Equations~\eqref{C2::2} and~\eqref{C2::3} hold if and only if
\begin{align*}
\sum_{G\in L(n, F,b)}c_G'M_2^{F,b}\left(\mathbb{O}_{\mathcal{N}},\mathbb{E}_{\mathcal{N}},\mathbb{U}_{\mathcal{N}}\right)&=0,\\
\sum_{G\in L(n, F,b)}c_G'M_3^{F,b}\left(\mathbb{E}_{\mathcal{N}},\mathbb{O}_{\mathcal{N}},\mathbb{U}_{\mathcal{N}}\right)&=0.
\end{align*}
Since $G_{\cdot 1}=(0,0)$, we know that
\begin{align*}
\text{ker}\left(M_1^{F,b}\left(\mathbb{E}_{\mathcal{N}},\mathbb{O}_{\mathcal{N}},\mathbb{U}_{\mathcal{N}}\right)\right)&=\,\text{ker}\left(M_3^{F,b}\left(\mathbb{E}_{\mathcal{N}},\mathbb{O}_{\mathcal{N}},\mathbb{U}_{\mathcal{N}}\right)\right),\\
\text{ker}\left(M_2^{F,b}\left(\mathbb{O}_{\mathcal{N}},\mathbb{E}_{\mathcal{N}},\mathbb{U}_{\mathcal{N}}\right)\right)&=\,\text{ker}\left(M_4^{F,b}\left(\mathbb{O}_{\mathcal{N}},\mathbb{E}_{\mathcal{N}},\mathbb{U}_{\mathcal{N}}\right)\right).
\end{align*}
\item[Case 2.2:] $1\notin F_1\cup F_2$ and $\beta =(0,1)$. Then, $\gamma=(1,0)$. Hence, similar to Case~2.1 we deduce from Lemma~\ref{lem:colors} that Equations~\eqref{C2::2} to~\eqref{C2::4} hold if and only if
\begin{align*}
\sum_{G\in L(n, F,b)}c_G'M_2^{F,b}\left(\mathbb{U}_{\mathcal{N}}^2,\mathbb{U}_{\mathcal{N}}^1,\mathbb{E}_{\mathcal{N}}\cup\mathbb{O}_{\mathcal{N}}\right)&=0,\\
\sum_{G\in L(n, F,b)}c_G'M_3^{F,b}\left(\mathbb{U}_{\mathcal{N}}^1,\mathbb{U}_{\mathcal{N}}^2,\mathbb{E}_{\mathcal{N}}\cup\mathbb{O}_{\mathcal{N}}\right)&=0,\\
\sum_{G\in L(n, F,b)}c_G'M_4^{F,b}\left(\mathbb{O}_{\mathcal{N}},\mathbb{E}_{\mathcal{N}},\mathbb{U}_{\mathcal{N}}\right)&=0.
\end{align*}
\item[Case 2.3:] $1\notin F_1$, $1\in F_2$ and $\beta=(0,0)$. Then, $\gamma=(1,0)$. Hence, similar to Case~2.1 we deduce from Lemma~\ref{lem:colors} that Equations~\eqref{C2::2} to~\eqref{C2::4} hold if and only if
\begin{align*}
\sum_{G\in L(n, F,b)}c_G'M_i^{F,b}\left(\mathbb{U}_{\mathcal{N}}^2,\mathbb{U}_{\mathcal{N}}^1,\mathbb{E}_{\mathcal{N}}\cup\mathbb{O}_{\mathcal{N}}\right)&=0~&~&\forall\,i\in\{2,4\},\\
\sum_{G\in L(n, F,b)}c_G'M_3^{F,b}\left(\mathbb{E}_{\mathcal{N}},\mathbb{O}_{\mathcal{N}},\mathbb{U}_{\mathcal{N}}\right)&=0.
\end{align*}
Since $G_{\cdot 1}=(0,0)$, we know that
\begin{align*}
\text{ker}\left(M_1^{F,b}\left(\mathbb{E}_{\mathcal{N}},\mathbb{O}_{\mathcal{N}},\mathbb{U}_{\mathcal{N}}\right)\right)&=\,\text{ker}\left(M_3^{F,b}\left(\mathbb{E}_{\mathcal{N}},\mathbb{O}_{\mathcal{N}},\mathbb{U}_{\mathcal{N}}\right)\right),\\
\text{ker}\left(M_2^{F,b}\left(\mathbb{U}_{\mathcal{N}}^2,\mathbb{U}_{\mathcal{N}}^1,\mathbb{E}_{\mathcal{N}}\cup\mathbb{O}_{\mathcal{N}}\right)\right)&=\,\text{ker}\left(M_4^{F,b}\left(\mathbb{U}_{\mathcal{N}}^2,\mathbb{U}_{\mathcal{N}}^1,\mathbb{E}_{\mathcal{N}}\cup\mathbb{O}_{\mathcal{N}}\right)\right).
\end{align*}
\item[Case 2.4:] $1\notin F_1$, $1\in F_2$ and $\beta=(0,1)$. Then, $\gamma=(1,1)$. Hence, similar to Case~2.1 we deduce from Lemma~\ref{lem:colors} that Equations~\eqref{C2::2} to~\eqref{C2::4} hold if and only if
\begin{align*}
\sum_{G\in L(n, F,b)}c_G'M_2^{F,b}\left(\mathbb{O}_{\mathcal{N}},\mathbb{E}_{\mathcal{N}},\mathbb{U}_{\mathcal{N}}\right)&=0,\\
\sum_{G\in L(n, F,b)}c_G'M_3^{F,b}\left(\mathbb{U}_{\mathcal{N}}^1,\mathbb{U}_{\mathcal{N}}^2,\mathbb{E}_{\mathcal{N}}\cup\mathbb{O}_{\mathcal{N}}\right)&=0,\\
\sum_{G\in L(n, F,b)}c_G'M_4^{F,b}\left(\mathbb{U}_{\mathcal{N}}^2,\mathbb{U}_{\mathcal{N}}^1,\mathbb{E}_{\mathcal{N}}\cup\mathbb{O}_{\mathcal{N}}\right)&=0.
\end{align*}
\item[Case 2.5:] $1\in F_1$, $1\notin F_2$ and $\beta=(0,0)$. Then, $\gamma=(0,1)$. Analogous to Case~2.3.
\item[Case 2.6:] $1\in F_1$, $1\notin F_2$ and $\beta=(1,0)$. Then, $\gamma=(1,1)$. Analogous to Case~2.4.
\end{description}

\item[Case 3:] $1\in F_1$, $b^1_1\neq 0$ and $1\in F_2$, $b^2_1\neq 0$. Let $\beta =(b_1^1,b_1^2)$. Then, similar to Case~2,
\begin{align*}
\psi_{\mathcal{N}}(f)&=\sum_{G\in L(n, F,b)}c_G\,\left(\prod_{j=1}^na_{G_{\cdot j}}^ja_{H_{\cdot j}}^j\right)\times\\
&~~~~\left(a_{\beta}^{n+1}a_{\beta}^{n+1}\prod_{j=2}^{n-1}a_{\sum_{l=1}^jG_{\cdot l}}^{n+j}a_{\sum_{l=1}^jH_{\cdot l}}^{n+j}+a_{\beta}^{n+1}a_{\beta}^{2n}\prod_{j=2}^{n-1}a_{\sum_{l=1}^jG_{\cdot l}}^{n+j}a_{\sum_{l=2}^jH_{\cdot l}}^{n+j}\right.\\
&~~~~~~\left.+a_{\beta}^{2n}a_{\beta}^{n+1}\prod_{j=2}^{n-1}a_{\sum_{l=2}^jG_{\cdot l}}^{n+j}a_{\sum_{l=1}^jH_{\cdot l}}^{n+j}+a_{\beta}^{2n}a_{\beta}^{2n}\prod_{j=2}^{n-1}a_{\sum_{l=2}^jG_{\cdot l}}^{n+j}a_{\sum_{l=2}^jH_{\cdot l}}^{n+j}\right)
\end{align*}
Since factors $a_{\beta}^{n+1}$ and $a_{\beta}^{2n}$ are distinct, all pairs of products in our sum can not cancel out except for possibly the second and third. In addition, for $j\in\mathbb{E}_{\mathcal{N}}\cup\mathbb{U}_{\mathcal{N}}$, we know that matrices
\begin{align*}
\left[\sum_{l=1}^jG_{\cdot l},\sum_{l=2}^jH_{\cdot l}\right]~~~\text{and}~~~\left[\sum_{l=2}^jG_{\cdot l},\sum_{l=1}^jH_{\cdot l}\right]
\end{align*}
are distinct under column permutations because $G_{\cdot 1}=H_{\cdot 1}\neq(0,0)$. If instead $\mathbb{E}_{\mathcal{N}}\cup\mathbb{U}_{\mathcal{N}}=\emptyset$, then
\begin{align*}
\sum_{l=1}^jG_{il}&\neq\sum_{l=1}^jH_{il}~&~&\forall\,i\in\{1,2\},j\in\{2,\dots,n-1\}.
\end{align*}
Hence, the second and third product in our sum can not cancel out. Therefore, like in Case~2, $\psi_{\mathcal{N}}(f)=0$ if and only if Equations~\eqref{C2::1} to~\eqref{C2::4} hold. Observe that Equations~\eqref{C2::1} and~\eqref{C2::4} as well as Equations~\eqref{C2::2} and~\eqref{C2::3} are equivalent because they only differ by an addition of $\beta$ to all subscripts of the Fourier parameters. 
Thus, similar to Case~2, our claim follows from Lemma~\ref{lem:colors}.
\end{description}


\bibliographystyle{siamplain}
\bibliography{manuscript}

\begin{thebibliography}{10}

\bibitem{Anfinsen59}
{\sc C.~B. Anfinsen}, {\em The molecular basis of evolution}, John Wiley \& Son
  Inc., New York, 1959.

\bibitem{Ardi21}
{\sc M.~Ardiyansyah}, {\em Distinguishing level-2 phylogenetic networks using
  phylogenetic invariants}, arXiv:2104.12479,  (2021).

\bibitem{Cavender78}
{\sc J.~A. Cavender}, {\em Taxonomy with confidence}, Mathematical Biosciences,
  40 (1978), pp.~271--280.

\bibitem{CF87}
{\sc J.~A. Cavender and J.~Felsenstein}, {\em Invariants of phylogenetics: a
  simple case with discrete states}, Journal of classification, 4 (1987),
  pp.~57--71.

\bibitem{CGH23}
{\sc J.~Cummings, E.~Gross, B.~Hollering, S.~Martin, and I.~Nometa}, {\em The
  pfaffian structure of cfn phylogenetic networks}, arXiv:2312.07450,  (2023).

\bibitem{CHM24}
{\sc J.~Cummings, B.~Hollering, and C.~Manon}, {\em Invariants for level-1
  phylogenetic networks under the cavendar-farris-neyman model}, Advances in
  Applied Mathematics, 153 (2024), p.~102633.

\bibitem{ES93}
{\sc S.~N. Evans and T.~P. Speed}, {\em Invariants of some probability models
  used in phylogenetic inference}, Annals of Statistics,  (1993).

\bibitem{EZ98}
{\sc S.~N. Evans and X.~Zhou}, {\em Constructing and counting phylogenetic
  invariants}, Journal of Computational Biology, 5 (1998), pp.~713--724.

\bibitem{Farris73}
{\sc J.~S. Farris}, {\em A probability model for inferring evolutionary trees},
  Systematic Zoology, 22 (1973), pp.~250--256.

\bibitem{Felsen04}
{\sc J.~Felsenstein}, {\em Inferring Phylogenies}, Sinauer Associates, Inc.,
  Sunderland, 2004.

\bibitem{FS95}
{\sc V.~Ferretti and D.~Sankoff}, {\em Phylogenetic invariants for more general
  evolutionary models}, Journal of Theoretical Biology, 173 (1995),
  pp.~147--162.

\bibitem{GIJ21}
{\sc E.~Gross, L.~van Iersel, R.~Janssen, M.~Jones, C.~Long, and Y.~Murakami},
  {\em Distinguishing level-1 phylogenetic networks on the basis of data
  generated by markov processes}, Mathematical Biology, 83 (2021), pp.~1--24.

\bibitem{Hage00}
{\sc T.~R. Hagedorn}, {\em Determining the number and structure of phylogenetic
  invariants}, Advances in Applied Mathematics, 24 (2000), pp.~1--21.

\bibitem{HS21}
{\sc B.~Hollering and S.~Sullivant}, {\em Identifiability in phylogenetics
  using algebraic matroids}, Journal of Symbolic Computation, 104 (2021),
  pp.~142--158.

\bibitem{JC69}
{\sc T.~H. Jukes and C.~R. Cantor}, {\em Evolution of protein molecules},
  Mammalian protein metabolism, 3 (1969), pp.~21--132.

\bibitem{K81}
{\sc M.~Kimura}, {\em Estimation of evolutionary distances between homologous
  nucleotide sequences}, Proceedings of the National Academy of Sciences, 78
  (1981), pp.~454--458.

\bibitem{Kong22}
{\sc S.~Kong, J.~C. Pons, L.~Kubatko, and K.~Wicke}, {\em Classes of explicit
  phylogenetic networks and their biological and mathematical significance},
  Journal of Mathematical Biology, 84 (2022).

\bibitem{Lake87}
{\sc J.~A. Lake}, {\em A rate-independent technique for analysis of nucleic
  acid sequences: evolutionary parsimony}, Molecular biology and evolution, 4
  (1987), pp.~167--191.

\bibitem{Long20}
{\sc C.~Long}, {\em Initial ideals of pfaffian ideals}, Journal of Commutative
  Algebra, 12 (202), pp.~91--105.

\bibitem{Neyman71}
{\sc J.~Neyman}, {\em Molecular studies of evolution: a source of novel
  statistical problems}, Statistical decision theory and related topics,
  (1971), pp.~1--27.

\bibitem{PS05}
{\sc L.~Pachter and B.~Sturmfels}, {\em Algebraic statistics for computational
  biology}, vol.~13, Cambridge university press, 2005.

\bibitem{Posada08}
{\sc D.~Posada}, {\em jmodeltest: phylogenetic model averaging}, Molecular
  biology and evolution, 25 (2008), pp.~1253--1256.

\bibitem{Steel16}
{\sc M.~Steel}, {\em Phylogeny: discrete and random processes in evolution},
  Society for Industrial and Applied Mathematics, 2016.

\bibitem{Steel03}
{\sc M.~Steel and C.~Semple}, {\em Phylogenetics}, Oxford University Press,
  Oxford, 2003.

\bibitem{SS05}
{\sc B.~Sturmfels and S.~Sullivant}, {\em Toric ideals of phylogenetic
  invariants}, Journal of Computational Biology, 12 (2005), pp.~457--481.

\bibitem{SS08}
{\sc B.~Sturmfels and S.~Sullivant}, {\em Toric geometry of cuts and splits},
  Michigan Mathematical Journal, 57 (2008), pp.~689--709.

\bibitem{Sull07}
{\sc S.~Sullivant}, {\em Toric fiber products}, Journal of Algebra, 316 (2007),
  pp.~560--577.

\bibitem{Yang14}
{\sc Z.~Yang}, {\em Molecular evolution: a statistical approach}, Oxford
  University Press, Oxford, 2014.

\end{thebibliography}

\end{document}